\documentclass[amsthm]{elsart}

\usepackage[pass]{geometry}

\usepackage{yjsco}

\makeatletter
  {\par\addvspace{\@bls \@plus 0.5\@bls \@minus 0.1\@bls}\noindent
   {\bfseries\Elproofname}\enspace\ignorespaces}%
  {\par\addvspace{\@bls \@plus 0.5\@bls \@minus 0.1\@bls}}
\makeatother
\newcommand*{\doi}[1]{\href{http://dx.doi.org/#1}{\texttt{doi: #1}}}

\usepackage[authoryear]{natbib}
\usepackage[utf8]{inputenc}
\usepackage{lmodern}
\usepackage{url}
\usepackage{mdwlist}
\usepackage{multicol}
\usepackage{multirow}
\usepackage{enumerate}
\usepackage{paralist}
\usepackage[utf8]{inputenc}
\usepackage{xspace}
\usepackage{amsmath,amssymb}
\usepackage{booktabs}
\usepackage{graphicx}
\graphicspath{{Pictures/}}   

\usepackage{algorithm}
\usepackage{algpseudocode}
\algrenewcommand\algorithmicrequire{\textbf{Input:}}
\algrenewcommand\algorithmicensure{\textbf{Output:}}
\algrenewcommand\algorithmicreturn{\textbf{Return}}
\algrenewcommand\Return{\State\algorithmicreturn{} }%

\newcommand{\fflasffpack}{\texttt{FFLAS-FFPACK}\xspace}
\newcommand{\RPM}[1]{\ensuremath{\mathcal{R}_{#1}\xspace}} 
\newcommand{\RS}[1]{\ensuremath{\text{RowSupp}({#1})\xspace}}  
\newcommand{\CS}[1]{\ensuremath{\text{ColSupp}({#1})\xspace}}  
\newcommand{\K}{\mathrm{K}\xspace}
\newcommand{\Z}{\mathbb{Z}\xspace}
\newcommand{\RRP}{\text{RowRP}}
\newcommand{\CRP}{\text{ColRP}}
\newcommand{\LA}{\ensuremath{\overline{A}_1\xspace}}

\newcommand{\lup}{\text{LUP}\xspace}
\newcommand{\MM}{\texttt{MM}\xspace}
\newcommand{\ple}{\text{PLE}\xspace}
\newcommand{\cupd}{\text{CUP}\xspace}
\newcommand{\plu}{\text{PLU}\xspace}
\newcommand{\pluq}{\text{PLUQ}\xspace}
\newcommand{\trsm}{\texttt{TRSM}\xspace}

\newcommand{\rank}{\text{rank}\xspace}
\newcommand{\submat}[3]{\ensuremath{{#1}_{#2,#3}}\xspace}
\newcommand{\subvec}[2]{\ensuremath{{#1}_{#2}}\xspace}
\DeclareMathOperator{\spanningrank}{s_r}
\DeclareMathOperator{\mccoyrank}{r_{\mathcal M}}

\newenvironment{smatrix}{\left[\begin{smallmatrix}}{\end{smallmatrix}\right]}

\newcommand{\breakalgorithm}{\algstore{bkbreak}\end{algorithmic}\end{algorithm}
\begin{algorithm}[htbp]
\begin{algorithmic}\algrestore{bkbreak}}

\usepackage{hyperref}
\hypersetup{
pdftitle={Fast Computation of the Rank Profile Matrix and the Generalized Bruhat Decomposition},
pdfauthor={Jean-Guillaume Dumas, Cl\'ement Pernet and Ziad Sultan},
breaklinks=true, 
colorlinks=true,
linkcolor=red,
 citecolor=blue,
 urlcolor=blue,
}
\begin{document}

\begin{frontmatter}

\title{Fast Computation of the Rank Profile Matrix and the Generalized Bruhat Decomposition}

\thanks{This research was partly supported by the HPAC project of the French
Agence Nationale de la Recherche (ANR 11 BS02 013) and the
 \href{http://opendreamkit.org/}{OpenDreamKit} \href{https://ec.europa.eu/programmes/horizon2020/}{Horizon 2020} \href{https://ec.europa.eu/programmes/horizon2020/en/h2020-section/european-research-infrastructures-including-e-infrastructures}{European Research Infrastructures} project (\href{http://cordis.europa.eu/project/rcn/198334_en.html}{\#676541})
.}

\author{Jean-Guillaume Dumas}
\address{Université Grenoble Alpes, Laboratoire LJK, umr CNRS, BP53X, 
  51, av. des Math\'ematiques, F38041 Grenoble, France}
\ead{jean-guillaume.dumas@imag.fr}
\ead[url]{http://www-ljk.imag.fr/membres/Jean-Guillaume.Dumas/}

\author{Clément Pernet}
\address{Université Grenoble Alpes, Laboratoire de l'Informatique du
  Parall\'elisme,  Université de Lyon, France.}
\ead{clement.pernet@imag.fr}
\ead[url]{http://lig-membres.imag.fr/pernet/}

\author{Ziad Sultan}
\address{Université Grenoble Alpes, Laboratoire LJK and LIG, Inria,
  CNRS, Inovall\'ee, 655, av. de l'Europe, F38334 St Ismier Cedex,
  France}
\ead{mailto:ziad.sultan@imag.fr}
\ead[url]{http://moais.imag.fr/membres/ziad.sultan}

\begin{abstract}
The row (resp.\ column) rank profile of a matrix describes the
stair-case shape of its row (resp.\ column) echelon form.
We describe a new matrix invariant, 
the rank profile matrix, summarizing all information on the row and
column rank profiles of all the leading sub-matrices.
We show that this normal form exists and is unique over a field but also over
any principal ideal domain and finite chain ring.
We then explore the conditions for a Gaussian
elimination algorithm to compute all or part of this invariant,
through the corresponding PLUQ decomposition. This enlarges the set of known
elimination variants that compute row or column rank profiles.
As a consequence a new Crout base case variant significantly improves the
practical efficiency of previously known implementations over a finite field.
With matrices of very small rank, we also generalize the techniques of
Storjohann and Yang to the computation of the rank profile matrix,
achieving an $(r^\omega+mn)^{1+o(1)}$ time complexity for an $m\times
n$ matrix of rank $r$, where $\omega$ is the exponent of matrix
multiplication. 
Finally, we give connections to the Bruhat decomposition, and several
of its variants and generalizations. Consequently, the algorithmic
improvements made for the PLUQ factorization, and their implementation,
directly apply to these decompositions.
In particular, we show how a PLUQ decomposition revealing the rank
profile matrix also reveals both a row and a column echelon form of
the input matrix or of any of its leading sub-matrices, by a simple
post-processing made of row and column permutations. 
\end{abstract}


\end{frontmatter}
\tableofcontents

\section{Introduction}

Triangular matrix decompositions are widely used in computational linear
algebra. Besides solving linear systems of equations, they are also
used to compute other objects more specific to exact arithmetic:
computing the rank,
sampling a vector from the null-space, computing echelon forms and
rank profiles.

The {\em row rank profile} (resp.\ {\em column rank profile}) of an $m\times n$ matrix $A$ with rank~$r$, denoted
by \RRP(A) (resp.\ \CRP(A)), is the
lexicographically smallest sequence of $r$ indices of linearly
independent rows (resp.\ columns) of $A$.
An $m\times n$ matrix has generic row (resp.\ column) rank profile if its row
(resp.\ column) rank profile is  $(1,..,r)$.
Lastly, an $m\times n$ matrix has generic rank profile if its $r$ first leading
principal minors are nonzero. Note that if a matrix has generic rank profile,
then its row and column rank profiles are generic, but the converse is false: the
matrix $\begin{smatrix}  0&1\\1&0\end{smatrix}$ does not have generic rank profile even if its row and column rank profiles
are generic.
The row support (resp.\ column support) of a matrix $A$, denoted by \RS{A}
(resp.\ \CS{A}), is the subset of indices of its nonzero rows (resp.\ columns). 

We recall that the row echelon form of an $m\times n$ matrix $A$ is an
upper triangular matrix $E=TA$, for a nonsingular matrix $T$,  with the zero rows
of $E$ at the bottom and the nonzero rows in stair-case shape:
$\min\{j:a_{i,j}\neq0\} < \min\{j:a_{i+1,j}\neq0\}$.
As $T$ is nonsingular, the column rank profile of $A$ is that of $E$, and
therefore corresponds to the column indices of the leading elements in the
staircase. Similarly the row rank profile of $A$ is composed of the row indices of the
leading elements in the staircase of the column echelon form of $A$.

\paragraph*{Rank profiles and triangular matrix decompositions.}
The rank profiles of a matrix and the triangular matrix decompositions obtained
by Gaussian elimination are strongly related.
The elimination of matrices with arbitrary rank profiles gives rise
to several matrix factorizations and many algorithmic variants.
In numerical linear algebra one often uses the \pluq decomposition, with
$P$ and $Q$ permutation matrices, $L$ a lower unit triangular matrix
and $U$ an upper triangular matrix.
The LSP and LQUP variants of~\cite{IMH:1982} have been introduced to reduce the
complexity of rank deficient Gaussian elimination to that of
matrix multiplication.
Many other algorithmic decompositions exist allowing fraction free computations
\cite{Jeffrey:2010:lufact}, in-place computations~\cite{jgd:2008:toms,JPS:2013}
or sub-cubic rank-sensitive time
complexity~\cite{Storjohann:2000:thesis,JPS:2013}.  
The reader may refer to~\cite{JPS:2013} for a detailed comparison between these
matrix factorizations, and further details on the \cupd (resp.\ \ple) variants, revealing
the row (resp.\ column) rank profiles.
All these algorithms, together with the schoolbook Gaussian elimination
algorithm share the property that, for a row rank profile computation, the pivot
search processes rows in order, and searches a pivot in all possible column position before
declaring the row linearly dependent with the previous ones.
As a consequence, blocking is limited to only one dimension (in this case
the row dimension) leading to slab algorithms~\cite{KlvdGe95} operating on rectangular
blocks of unbalanced dimensions.
This reduces the data locality of the algorithm, and therefore
penalizes the efficiency of implementations in practice. In parallel, this
blocking also puts more constrains on the dependencies between tasks~\cite{DPS15}.

\paragraph*{Contribution with respect to the state of the art.}
In \cite{DPS:2013} we proposed a first Gaussian elimination algorithm,
with a recursive splitting of both row and column dimensions, which
simultaneously computes the row and column rank profile while
preserving the sub-cubic rank-sensitive time complexity and keeping
the computation in-place. It showed that slab blocking is not a
necessary condition for a Gaussian elimination to reveal rank
profiles. Consequently, we have further analyzed the conditions on the
pivoting that reveal the rank profiles in~\cite{DPS:ELU:2015}, where
we introduced a new matrix invariant, the rank profile matrix. This
normal form contains the row and column rank profile information of
the matrix and that of all its leading sub-matrices.  

This normal form is closely related to a permutation matrix appearing
in the Bruhat decomposition~\cite{Bruhat:1956:Lie} and in related
variants~\cite{delladora:1973:these,Grigoriev:1981:bruhat,Bourbaki:2008:lie,Malaschonok:2010,Manthey:2007:Bruhat}.
Still, none of these did connect it to the notion of rank profile. In another
setting, the construction of matrix Schubert 
varieties in~\cite[Ch. 15]{MillerSturmfels:2005} defines a similar
invariant, but presents neither a matrix decomposition nor any
computational aspects. 

More precisely, the present paper gathers the key contributions
of~\cite{DPS:2013} and~\cite{DPS:ELU:2015}:
\begin{enumerate}
\item we define a new matrix
invariant over a field, the rank profile matrix, summarizing all
information on the row and column rank profiles of all the leading
sub-matrices;
\item we study the conditions for a Gaussian elimination algorithm to compute
  all or part of this invariant, through the corresponding PLUQ
  decomposition;
\item as a consequence, we show that the classical iterative CUP
decomposition algorithm can actually be adapted to compute the rank
profile matrix. Used, in a Crout variant, as a base-case to our
recursive implementation over a finite field, it delivers a significant improvement in efficiency;
\item we also show that both the row and the column echelon forms of a
  matrix can be recovered from some PLUQ decompositions thanks to an
  elementary post-processing algorithm.
\end{enumerate}

Further, we develop three novel aspects:
\begin{enumerate}
\item we study how the notion of rank profile matrix can be generalized over an arbitrary
ring. We show that none of the common definition of rank over an arbitrary
commutative ring allow to define a rank profile matrix in general. However, over
a principal ideal domain and over a finite chain ring, we can produce such a definition.
\item we make further connections with existing matrix
decompositions, in particular the Bruhat and the generalized Bruhat
decompositions, exhibiting both the row and the column echelon in a single decomposition;
\item lastly, we extend the recent algorithmic improvements 
of~\cite{CKL:13,StoYan14,StoYan15} for low rank matrices: 
indeed, we show here that the algorithm in~\cite{StoYan14} 
computes the rank profile matrix;  propose an algorithmic variant thereof,
reducing the leading constant by a factor of three; and lastly show an
algorithmic reduction computing the rank profile matrix in time
bounded by $(r^\omega+mn)^{1+o(1)}$ with a Las-Vegas probabilistic algorithm.
\end{enumerate}

\paragraph*{Organization of the article.}
We first introduce in Section~\ref{sec:rkp} the rank profile matrix $\RPM{A}$,
and study in Section~\ref{sec:rpm:ring} its generalization over an arbitrary
ring.
We then study in Section~\ref{sec:WhenPLUQ} under which condition a \pluq
decomposition algorithms reveals the rank profile
structure of a matrix. For this, we investigate existing and new pivoting
strategies, based on all combination of elementary search and permutation
operations, showing for each of them what part of the rank profile information is being computed.
In particular we propose three new pivoting strategies that compute the rank profile
matrix.
As an illustration, we show in Section~\ref{sec:algo} how these pivoting
strategies instantiate in iterative or recursive algorithms, using slab or tile
blocking. Connections are made to the most common elimination algorithms and we
state in full details the recent tile recursive algorithm of~\cite{DPS:2013},
implementing one of the new pivoting strategy. 
Section~\ref{sec:basecase} shows how this better understanding on the pivoting strategies
has resulted in the design of an iterative Crout \cupd decomposition with rotations,
to be used as a base case for the tile recursive algorithm, speeding up the
computation efficiency, while still recovering the rank profile matrix
invariant.
We then show in Section~\ref{sec:relations} how a \pluq decomposition revealing
the rank profile matrix relates with other triangular matrix decompositions,
such as the LEU and the classical, modified or generalized Bruhat
decompositions, or the computation of both row and
column echelon forms, from a single \pluq decomposition.
Lastly, we extend in Section~\ref{sec:lowrank} the recent algorithms
of~\cite{StoYan14,StoYan15} for the row or column rank profile of matrices with low
rank to computed the rank profile matrix within the same complexities.

\paragraph*{Notations.}
In the following, $0_{m\times n}$ denotes the $m\times n$ zero matrix. For two
list of indices $\mathcal{P}$ and $\mathcal{Q}$,
$\submat{A}{\mathcal{P}}{\mathcal{Q}}$ denotes the sub-matrix of $A$ formed by
the rows of index in $\mathcal{P}$ and the columns of index in $\mathcal{Q}$.
In particular, $A_{i..j,k..l}$ denotes the contiguous block of coefficients in $A$
of rows position between $i$ and $j$ and columns position between $k$ and
$l$. We may denote by $*$ the sequence of all possible row or column indices:
e.g. $\submat{A}{i}{*}$ denotes the $i$-th row of $A$.
To a permutation $\sigma:\{1,\dots,n\}\rightarrow \{1,\dots,n\}$ we define the associated
permutation matrix $P(\sigma)$, permuting rows by left multiplication: the rows
of $P(\sigma) A$ are that of $A$ permuted by $\sigma$. Reciprocally, for a
permutation matrix $P$, we denote by $\sigma(P)$ the associated permutation.

\section{The rank profile matrix}
\label{sec:rkp}

We propose in Theorem~\ref{thm:rankprofilematrix} the definition of the rank profile
matrix, an invariant summarizing all information on the rank profiles of a
matrix. As will be discussed in this section and in Section~\ref{sec:relations},
this invariant is closely related to the Bruhat decomposition~\cite{Bruhat:1956:Lie} and its
generalizations~\cite{Grigoriev:1981:bruhat,Tyrtyshnikov:1997:Bruhat,MillerSturmfels:2005}.

\subsection{Definition over a field}
We first consider matrices over an arbitrary commutative field $\K$.

\begin{defn} An $r$-sub-permutation matrix is a matrix of rank $r$ with only  $r$
  non-zero entries equal to one.
\end{defn}

\begin{lem}
An $m\times n$ $r$-sub-permutation matrix
has at most one non-zero entry per row and per column,
and
can be written $P
  \begin{smatrix}
    I_r\\&0_{(m-r)\times (n-r)}
  \end{smatrix}Q$ where $P$ and $Q$ are permutation matrices.
\end{lem}

\begin{thm}\label{thm:rankprofilematrix}
 Let $A\in\K^{m\times n}$ of rank $r$. There exists a unique $m\times n$
 $r$-sub-permutation matrix $\RPM{A}$ 
 of which every leading sub-matrix has the same
 rank as the corresponding
 leading sub-matrix of $A$. 
 This sub-per\-mu\-ta\-tion matrix is called the \em{rank profile matrix} of~$A$.
\end{thm}


\begin{pf}
  We prove existence by induction  on the row dimension.
 of the leading  submatrices.

  If $A_{1,1..n} = 0_{1\times n}$, setting $\RPM{}^{(1)} = 0_{1\times n}$ satisfies the
  defining condition.
  Otherwise, let $j$ be the index of the first invertible element in
  $A_{1,1..n}$ and set $\RPM{}^{(1)}= e_j^T$ the j-th $n$-dimensional canonical row
  vector, which satisfies the defining condition.

  Now for a given $i\in\{1,\dots,m\}$, suppose that there is a unique $i\times n$
  rank profile matrix $\RPM{}^{(i)}$ such that $\rank(A_{1..\ell,1..j}) =
  \rank(\RPM{1..\ell,1..j})$ for every $j\in\{1..n\}$ and $\ell\in\{1..i\}$. 
  If $\rank(A_{1..i+1,1..n})=\rank(A_{1..i,1..n})$, then $\RPM{}^{(i+1)}=
  \begin{smatrix}
    \RPM{}^{(i)}\\0_{1\times n}
  \end{smatrix}$. 
  Otherwise, consider $k$, the smallest column index such that
  $\rank(A_{1..i+1,1..k})=\rank(A_{1..i,1..k})+1$ and set $\RPM{}^{(i+1)}=
  \begin{smatrix}
    \RPM{}^{(i)}\\e_k^T
  \end{smatrix}$. Define $
  \begin{smatrix}
    B & u \\
    v^T & x 
  \end{smatrix} = A_{1..i+1,1..k}$,
 where $u,v$ are vectors and $x$ is a scalar.
By definition of $k$, we have $\rank(B) = \rank(
\begin{smatrix}  B\\v^T\end{smatrix})$.

First we show that $\RPM{}^{(i+1)}$ is an $r_{i+1}$-sub-per\-mu\-ta\-tion
matrix. 
If not, then the $k$-th column of $\RPM{}^{(i)}$ would contain a $1$ which, by
induction hypothesis, would imply that
$\rank(\begin{bmatrix}B&u\end{bmatrix})=\rank(B)+1$.
Hence we would have $\rank(\begin{smatrix}  B&u\\v^T&x\end{smatrix}) 
= \rank(\begin{bmatrix}  B&u\end{bmatrix}) + 1 = \rank(B)+2 = \rank(
\begin{smatrix}  B\\v^T\end{smatrix})+2$, a contradiction.

Then we show that any leading sub-matrix $A_{1..s,1..t}$ of $A$ with $s\leq i+1,
t\leq n$ has the same rank as the corresponding leading sub-matrix of $\RPM{}^{(i+1)}$. 
The case $s\leq i$ is covered by the induction; 
second,  since  $\rank(B) = \rank(\begin{smatrix}  B\\v^T\end{smatrix})$, 
any leading sub-matrix of $\begin{smatrix}B\\v^T\end{smatrix}$
has the same rank as the corresponding sub-matrix of
$\RPM{}^{(i+1)}$, which covers the case $t<k$.
Lastly, for $s=i+1$ and $t\geq k$, the definition of $k$ implies
 $\rank(A_{1..s,1..t}) =\rank(A_{1..i,1..t})+1 =
 \rank(\RPM{1..i,1..t}^{(i)})+1 =  \rank(\RPM{1..i,1..t}^{(i+1)})$.

  To prove uniqueness, suppose there exist two distinct rank profile
  matrices $\RPM{}^{(1)}$ and $\RPM{}^{(2)}$ for a given matrix $A$ and let
  $(i,j)$ be the lexicographically minimal coordinates where
  $\RPM{i,j}^{(1)}\neq\RPM{i,j}^{(2)}$. The rank of the
  $(i,j)$-leading submatrices of $\RPM{}^{(1)}$ and $\RPM{}^{(2)}$ differ but should both
  be equal to $\rank(A_{1..i,1..j})$, a contradiction.
\end{pf}

\begin{exmp}$A= 
  \begin{smatrix}
    2 & 0 & 3& 0\\
    1 & 0 & 0& 0\\
    0 & 0 & 4 &0\\
    0 & 2 & 0 &1\\
  \end{smatrix}
  $ has $\RPM{A} = 
  \begin{smatrix}
    1& 0& 0& 0\\
    0& 0& 1& 0 \\
    0& 0& 0& 0\\
    0& 1& 0& 0
  \end{smatrix}
  $ for rank profile matrix over $\mathbb{Q}$.
\end{exmp}

\begin{rem}
  The permutation matrix introduced in the {\em modified Bruhat
    decomposition} of~\cite{Tyrtyshnikov:1997:Bruhat}, and defined
  there only for invertible matrices, is also the matrix $E$
  introduced in Malaschonok's LEU
  decomposition~\cite[Theorem~1]{Malaschonok:2010}.
  In the latter paper, an algorithm for this decomposition was
  only shown over a field for $m=n=2^k$, and no 
  connection was made to the relation with ranks and rank profiles.
  We have shown in~\cite[Corollary~1]{DPS:2013}  that $E$ is in fact the rank profile
  matrix and made the connection to the PLUQ decomposition explicit, as recalled
  in Section~\ref{sec:relations}. 
  We here generalize the existence to arbitrary rank $r$ and
  dimensions $m$ and $n$ and after proving its
  uniqueness, we propose this definition as a new matrix normal
  form. 
 \end{rem}


The rank profile matrix has the following properties:
\begin{lem}\label{lem:rpm:prop} Let $A$ be a matrix.
\newcounter{myenum}
  \begin{enumerate}
  \item   $\RPM{A}$ is {\em diagonal} if and only if $A$ has {\em generic rank profile}.
  \item $\RPM{A}$ is a {\em permutation} matrix if and only if $A$ is  {\em invertible} 
  \item $\RRP(A) = \RS{\RPM{A}}$; $\CRP(A) = \CS{\RPM{A}}$.
\setcounter{myenum}{\theenumi}
\end{enumerate}
Moreover, for all $1\leq i\leq m$ and $1\leq j\leq n$, we have:
\begin{enumerate}\setcounter{enumi}{\themyenum}
  \item $\RRP(A_{1..i,1..j}) = \RS{(\RPM{A})_{1..i,1..j}}$ 
  \item $\CRP(A_{1..i,1..j}) = \CS{(\RPM{A})_{1..i,1..j}}$,
  \end{enumerate}
\end{lem}
These properties show how to recover the row and column rank profiles
of $A$ and of any of its leading sub-matrix.

\subsection{Generalization over a commutative ring with unity}\label{sec:rpm:ring}

\newcommand{\Ring}{\ensuremath{\mathrm{R}}\xspace}

We now explore if and under which  condition, the notion of rank profile matrix
can be generalized over an arbitrary commutative ring with unity. 
As the rank profile matrix over a field relies on the notion
of rank, which generalization over an arbitrary ring leads to several
definitions, we review the most common of these definitions and explore their
capacity to define a rank profile matrix invariant.

Let $\Ring$ be a commutative ring with unity.
Following~\cite[\S 4]{Brown:1992:matring}, we will denote by $I_t(A)$ the $t$-th
determinantal ideal, namely the ideal in $\Ring$ generated by all 
$t\times t$ minors of $A$, and set by definition $I_0=\Ring$. We have $\Ring =I_0 \supseteq I_1(A) \supseteq I_2(A)
\supseteq \dots \supseteq (0)$.
We will also denote by $\text{Ann}(S) = \{x\in \Ring : \forall y \in S, xy=0\}$ the
annihilator ideal of a set $S$.

\begin{defn}\label{def:ringrank} For an $m\times n$ matrix $A$ over a communtative ring $\Ring$,
  the rank of $A$ can be defined as:

  \begin{enumerate}
  \item $r_\text{Span} = \min \{t : A=BC \text{ where } B \text{ is } m\times
    t \text{ and } C \text{ is } t\times n\}$, called spanning rank
    by~\cite{Brown:1998:spanningrank} or Schein rank by \cite[\S 2.4]{Rao:2002},
  \item $r_\text{McCoy} = \max\{t : \text{Ann}(I_t(A))=(0)\}$ introduced
    in~\cite[Theorem~51]{McCoy:1948:rings},
  \item $r_\text{Det} = \max\{t : A \text{ has a nonzero } t\times t \text{
      minor}\}$ \cite[\S 2.4]{Rao:2002} (this is also $\max\{t : I_t(A)\neq (0)\}$ defined in~\cite[\S 4, Ex. 11]{Brown:1992:matring}),
  \item $r_\text{UnitMinor} = \max\{t : A \text{ has a unit } t\times t \text{
      minor}\}$ used in~\cite[Def.~2.5]{Norton:2000:fcr}.
  \end{enumerate}
\end{defn}

We illustrate the value of these ranks in Table~\ref{tab:ringranks} for three
$2\times 2$ matrices (columns~3, 6, 9) and their leading $1\times 2$ and 
$2\times 1$ sub-matrices (columns~2, 5, 8). These ranks will be used to argue
whether a rank profile matrix can be defined in each case (columns~4, 7 and 10 if it exists).
\begin{table}[h]
  \centering
\setlength{\tabcolsep}{.88\tabcolsep}
  \begin{tabular}{l|ccc|ccc|ccc}
    \toprule
&\multicolumn{3}{c|}{over $\Z/4\Z$}& \multicolumn{6}{c}{over $\Z/6\Z$}\\
& $\begin{smatrix} 0 & 2\end{smatrix}$, $\begin{smatrix} 0 \\ 2\end{smatrix}$ &$\begin{smatrix} 0 & 2 \\ 2 & 1\end{smatrix}$
& $\RPM{  \begin{smatrix}    0&2\\2& 1  \end{smatrix}}$ 
& $\begin{smatrix}2& 3\end{smatrix}$, $\begin{smatrix}2\\ 3\end{smatrix}$  
& $\begin{smatrix} 2 & 3 \\ 3 & 0\end{smatrix}$
& $\RPM{  \begin{smatrix}    2&3\\3& 0  \end{smatrix}}$ 
& $\begin{smatrix}2\\1\end{smatrix}$, $\begin{smatrix}2&3\end{smatrix}$ 
& $\begin{smatrix} 2 & 3 \\ 1 & 1\end{smatrix}$
& $\RPM{  \begin{smatrix}   2&3\\1& 1  \end{smatrix}}$ 
\\
    \midrule
$r_\text{Span}$&  1 & 1 &none&1 &  2 &$\begin{smatrix}1&0\\0&1\end{smatrix}$& 1 & 2 &$\begin{smatrix} 1&0\\0&1\end{smatrix}$
\\
$r_\text{McCoy}$&  0 & 1 &$\begin{smatrix}0&0\\0&1\end{smatrix}$&1 & 1 & none &1
  &2& 
  $\begin{smatrix}0&1 \\ 1&0  \end{smatrix}$
  \\
$r_\text{Det}$&  1 & 1 &none &1  & 2 & $\begin{smatrix}1&0\\0&1\end{smatrix}$ &1 & 2 & $\begin{smatrix} 1&0\\0&1\end{smatrix}$ \\
$r_\text{UnitMinor}$&  0 & 0 &$\begin{smatrix} 0&0\\0&0\end{smatrix}$&0 &  0 & $\begin{smatrix}0&0\\0&0\end{smatrix}$&1,0 &2 & none\\
\bottomrule
  \end{tabular}

  \caption{Ranks of matrices according to the various definitions. These three
    $2\times 2$ matrices are counterexamples showing that none of the four rank
    definitions of Definition~\ref{def:ringrank} can be used to define a rank profile
    matrix in general.}
\label{tab:ringranks}
\end{table}

For instance $r_\text{Span}(\begin{smatrix} 0&2\\2&1\end{smatrix}) = 1$ since $\begin{smatrix} 0&2\\2&1\end{smatrix}=
\begin{smatrix} 2\\1\end{smatrix}
\begin{smatrix} 2&1\end{smatrix}
$ over $\Z/4\Z$. 
This matrix shows that $r_\text{Span}$ can not be used to define a rank profile
matrix. Indeed, if it would exist, its first row would be $\begin{smatrix}
  0&1\end{smatrix}$ (as the $1\times 1$ leading submatrix has rank 0 and the
$1\times 2$ leading submatrix has rank 1). Similarly the first column of this
rank profile matrix would be $\begin{smatrix}  0&1\end{smatrix}^T$. But the rank
of the permutation matrix $\begin{smatrix}  0&1\\1&0\end{smatrix}$ would be
$2\neq r_\text{Span}(\begin{smatrix}  0&2\\2&1\end{smatrix})$.
We will use the same reasoning pattern, to show that $r_\text{Det}$ can not define
a rank profile matrix for $\begin{smatrix}  0&2\\2&1\end{smatrix}$, and neither
does $r_\text{McCoy}$ for $\begin{smatrix}  2&3\\3&0\end{smatrix}$ over $\Z/6\Z$: in these two
cases, the rank profile matrix should be $
\begin{smatrix}  0&1\\1&0\end{smatrix}$ to satisfy the rank condition on the
$2\times 1$  and $1\times 2$ sub-matrices, but the $2\times 2$ rank is only 1.
Laslty, $r_\text{UnitMinor}(\begin{smatrix}  2&3\end{smatrix})=0$ and
$r_\text{UnitMinor}(\begin{smatrix}2\\ 1\end{smatrix})=1$, therefore the rank
profile matrix of $\begin{smatrix}  2&3\\1&1\end{smatrix}$
should be of the form $\begin{smatrix}  0&0\\1&*\end{smatrix}$,
but it can then never be of rank 2, as the matrix~$\begin{smatrix} 2&3\\1&1\end{smatrix}$. 

Remark also that the rank profile matrix invariant is strongly connected with
elimination as will 
be presented in the next sections. It therefore needs to be based on a notion of
rank that is stable with multiplication by invertible matrices. This is the case
with McCoy's rank~\cite[4.11.(c)]{Brown:1992:matring}, but not with
$r_\text{UnitMinor}$. Indeed the rank of $\begin{smatrix} 2&1&\end{smatrix} =
\begin{smatrix} 2&3\end{smatrix} \begin{smatrix}  1&-1\\0&1\end{smatrix}$ is 1
whereas the rank of $\begin{smatrix}  2&3\end{smatrix}$ is 0.

Consequently there is no notion of rank over an arbitrary commutative ring with
unity supporting the definition of a rank profile matrix. We will now show that with additional
assumptions on the ring, some of the above definitions for the rank coincide and
make the existence of the rank profile matrix possible.

\subsubsection{Over a principal ideal domain}

Over a principal ideal domain, the existence of an underlying field of fractions
guarantees that the McCoy rank, the spanning rank and the rank over the field of fractions
coincide.
\newcommand{\Domain}{\textrm{D}\xspace}
\begin{cor} Let $\Domain$ be a principal ideal domain (PID) and let
  $A\in\Domain^{m\times{}n}$ with McCoy's rank $r$.  
 There exists a unique $m\times n$ $r$-sub-permutation matrix $\RPM{A}$ 
 of which every leading sub-matrix has the same
 rank as the corresponding
 leading sub-matrix of $A$. 
 This sub-per\-mu\-ta\-tion matrix is called the \em{rank profile
   matrix} of $A$.
\end{cor}
\begin{pf}
From~\cite[Proposition 1.6]{Brown:1998:spanningrank}, over a PID with
field of fractions $\K$, $\spanningrank=\mccoyrank=\rank_\K$. Thus
$\RPM{A}$ over $\K$ satisfies the requirements over $\Domain$ and is the
unique such matrix.
\end{pf}

\subsubsection{Over a finite chain ring}




A finite chain ring is a finite commutative ring with identity which ideals are ordered by
inclusion~\citep{ClLi73}. Equivalently, it is a finite local ring which maximal
ideal is principal. These rings include $\Z/p^k\Z$ and $\text{GF}(p^k)$ for $p$
prime, as well as all Galois rings of characteristic $p^k$.


Over a finite chain ring, McCoy rank and the UnitMinor rank coincide, and therefore
allow to define a rank profile matrix, as shown next.

\begin{lem}\label{lem:mccoyeq}
Over a finite chain ring, McCoy's rank is the largest positive  integer  
$r\leq\{m,n\}$ such that there exist a unit $r\times{}r$
minor of A. 
\end{lem}

\begin{pf}
Krull's theorem~\cite[Theorem~2]{Krull1938} (see also
\cite[Proposition~4]{hungerford1968}) states that in a local ring, an
element is a unit if and only if it is not in the maximal ideal.
As the maximal ideal of the ring is principal, let $g$ be a generator thereof.
Then $g$ is nilpotent.
Indeed the ring is finite so there exists indices $i<j$ such that
$g^i=g^j$. Therefore $g^i(1-g^{j-i})=0$. But as $g$ is not a unit,
$g^{j-i}$ is not a unit either, but then $1-g^{j-i}$ must be a unit
(otherwise the maximal ideal would contain $1$ and would not be proper). 
Therefore $g^i=0$.
The nilpotency index of $g$ is then the smallest positive
integer $\nu$ such that $g^\nu=0$.

Let $r=r_\text{McCoy}(A)$ and suppose that all $r\times r$ minors are
non-unit. They belong to the maximal ideal and $I_r(A) = (g^i)$ for some $1\leq
i < \nu$. Therefore, $\text{Ann}(I_r(A)) = (g^{\nu-i})$, a contradiction.
\end{pf}

\begin{cor}[{\cite[Corollary~2.7]{Norton:2000:fcr}}]\label{cor:unitmccoy}
Over a finite chain ring $\Ring$, with maximal ideal $(g)$, McCoy's rank is the rank over the field $R/gR$.
\end{cor}

\begin{cor} Let $\Ring$ be a finite chain ring and let
  $A\in\Ring^{m\times{}n}$ with McCoy's rank $r$.  
 There exists a unique $m\times n$ $r$-sub-permutation matrix $\RPM{A}$ 
 of which every leading sub-matrix has the same
 McCoy's rank as the corresponding
 leading sub-matrix of $A$. 
 This sub-per\-mu\-ta\-tion matrix is called the \em{rank profile
   matrix} of $A$.
\end{cor}
\begin{pf}
From Corollary~\ref{cor:unitmccoy}, we consider the residue field
$\K=\Ring/g\Ring$. Then
$\RPM{A}$ over $\K$ satisfies the requirements over $\Ring$ and is the
unique such matrix.
\end{pf}



\section{When does a PLUQ algorithm reveal the rank profile matrix?}
\label{sec:WhenPLUQ}
From now on, for the sake of simplicity, we consider algorithms over a field.
\subsection{Ingredients of a \pluq decomposition algorithm}
\label{sec:structurePLUQ}
Over a field, the LU 
decomposition generalizes to matrices with arbitrary rank
profiles, using row and column permutations   (in some cases such as
the CUP, or LSP decompositions, the row permutation is embedded in the structure of the $C$ or $S$
matrices). 
However such \pluq decompositions are not unique and not all of them will
necessarily reveal rank profiles and echelon forms.
We will characterize the conditions for a \pluq decomposition algorithm to
reveal the row or column rank profile or the rank profile matrix.

We consider the four types of operation of a Gaussian elimination
algorithm in the processing of the $k$-th pivot:
\begin{description}
  \item[Pivot search:] finding an element to be used as a pivot,
  \item[Pivot permutation:] moving the pivot in diagonal position $(k,k)$  by column and/or row permutations,
  \item[Update:] applying the elimination at position $(i,j)$: 
$a_{i,j}\leftarrow a_{i,j} -a_{i,k}a_{k,k}^{-1}a_{k,j}$,
  \item[Normalization:] dividing the $k$-th row (resp.\ column) by the pivot.
\end{description}
Choosing how each of these operation is done, and when they are scheduled
results in an elimination algorithm.  Conversely, any Gaussian elimination
algorithm computing a \pluq decomposition can be viewed as a set of
specializations of each of these operations together with a scheduling. 


The choice of doing the normalization on rows or columns only determines which
of $U$ or $L$ will be unit triangular. The scheduling of the updates vary
depending on the type of algorithm used: iterative, recursive, slab or tiled block splitting,
with right-looking, left-looking or Crout variants (see~\cite{DDSV98}).
Neither the normalization nor the update impact the capacity to reveal rank
profiles and we will thus now focus on the pivot search and permutation. 

Choosing a search and a permutation strategy sets the
 matrices $P$ and  $Q$ of the \pluq decomposition obtained and, as we
will see, determines the ability to recover information on the rank profiles.
Once these  matrices are fixed, the $L$ and the $U$ factors are unique.
We therefore introduce the pivoting matrix.
\begin{defn}\label{def:PivMat}
The pivoting matrix of a \pluq decomposition $A=PLUQ$ of rank $r$ is the $r$-sub-permutation matrix
$$\Pi_{P,Q}=P
\begin{smatrix}
  I_r\\&0_{(m-r)\times (n-r)}
\end{smatrix}
Q
.$$
\end{defn}
The $r$ nonzero elements of $\Pi_{P,Q}$ are  located at the initial positions of
the pivots in the matrix $A$. Thus $\Pi_{P,Q}$ summarizes the choices
made in the search and  permutation operations.

\begin{subsubsection}{Pivot search}
%
The search operation vastly differs depending on the field of application. In
numerical dense linear algebra, numerical stability is the
main criterion for the selection of the pivot. In sparse linear algebra, the pivot
is chosen so as to reduce the fill-in produced by the update operation.
In order to reveal some information on the rank profiles, a notion of precedence
has to be used: a usual way to compute the row rank profile is to search in a
given row for a pivot and only move to the next row if the current row was found to be all
zeros. This guarantees that
each pivot will be on the first linearly independent row, and therefore the row
support of $\Pi_{P,Q}$ will be the row rank profile.
The precedence here is that the pivot's coordinates must minimize the order
for the first coordinate (the row index). 
As a generalization, we consider the most common preorders of the cartesian product $\{1,\ldots
m\}\times \{1,\ldots n\}$ inherited from the natural orders of each of its
components and describe the corresponding search strategies, 
minimizing this preorder: 
\begin{description}
  \item[Row order:] $(i_1,j_1)\preceq_{\text{row}} (i_2,j_2)$ iff $i_1\leq i_2$:
    {\em search for any invertible element in the first nonzero row.}
  \item[Column order:] $(i_1,j_1)\preceq_{\text{col}} (i_2,j_2)$ iff $j_1\leq
    j_2$. 
    {\em search for any invertible element in the first nonzero column.}
  \item[Lexicographic order:] $(i_1,j_1)\preceq_{\text{lex}} (i_2,j_2)$ iff $i_1<i_2$ or $i_1=i_2$ and
    $j_1 \leq j_2$:
    {\em search for the leftmost nonzero element of the first nonzero row.}
  \item[Reverse lexicographic order:] $(i_1,j_1)\preceq_{\text{revlex}} (i_2,j_2)$ iff $j_1<j_2$ or
    $j_1=j_2$ and $i_1 \leq i_2$: {\em search for the topmost  nonzero element
      of the first nonzero column.}
  \item[Product order:]\index{product order} $(i_1,j_1)\preceq_{\text{prod}} (i_2,j_2)$ iff
    $i_1\leq i_2$ and $j_1\leq j_2$:
    {\em search for any nonzero element at position $(i,j)$ being the
      only nonzero of the leading $(i,j)$ sub-matrix.}
\end{description}
\begin{exmp}
Consider the matrix
$
\begin{smatrix}
  0 & 0 & 0 & a & b\\
  0 & c & d & e & f\\
  g & h & i & j & k\\
  l & m & n & o & p
\end{smatrix}$, where each literal  is a nonzero element.
The nonzero elements minimizing each preorder are the following:

\begin{center}
  \begin{tabular}{llll}
Row order &
 $a,b$ &
Column order&
$g,l$\\
Lexicographic order&
$a$ &
Reverse lexic. order&
$g$\\
Product order&
$a,c,g$\\
\end{tabular}
\end{center}
\end{exmp}


\end{subsubsection}

\begin{subsubsection}{Pivot permutation}

The pivot permutation moves a pivot from its initial position to the
leading diagonal. Besides this constraint all possible choices are left for the
remaining values of the permutation. 
Most often, it is done by row or column transpositions, as it clearly involves
a small amount of data movement.
However, 
these transpositions can break the precedence relations in the set of rows or
columns, and can therefore prevent the recovery of the rank profile information.
A pivot permutation that leaves the precedence relations unchanged will be
called  $k$-monotonically increasing.
\begin{defn}
  A permutation of $\sigma \in \mathcal{S}_n$ is called
  $k$-mono\-ton\-i\-cal\-ly increasing if its last $n-k$ values 
  form a monotonically increasing sequence.
\end{defn}
In particular, the last $n-k$ rows of the associated
row-permutation matrix $P_\sigma$
are in row echelon form.
For example, the cyclic shift between indices $k$ and $i$, with $k<i$
defined as $R_{k,i}=(1,\ldots,k-1,i,k,k+1,\ldots,i-1,i+1,\ldots,n)$, that we will call a
$(k,i)$-rotation, is an elementary $k$-monotonically increasing permutation.
\begin{exmp} The $(1,4)$-rotation $R_{1,4}=(4,1,2,3)$ is a
  $1$-mono\-to\-ni\-cal\-ly increasing permutation. Its row permutation matrix
  is 
$\begin{smatrix}
    0& & & 1\\
    1&    &   & \\
    &1& & \\
    & & 1&0\\
  \end{smatrix}$. 
\end{exmp}

Monotonically increasing permutations can be composed as stated in
Lemma~\ref{lem:permcompo}. 
\begin{lem}\label{lem:permcompo}
If $\sigma_1 \in \mathcal{S}_n$ is a $k_1$-monotonically increasing permutation
and $\sigma_2\in \mathcal{S}_{k_1} \times \mathcal{S}_{n-k_1}$ a
$k_2$-monotonically increasing permutation with $k_1<k_2$ then the
permutation $\sigma_2 \circ \sigma_1$ is a $k_2$-monotonically increasing
permutation. 
\end{lem}

 \begin{pf}
   The last $n-k_2$ values of $\sigma_2 \circ \sigma_1$ are the image of a
   sub-sequence of $n-k_2$ values from the last $n-k_1$ values of $\sigma_1$
   through the  monotonically increasing function~$\sigma_2$.
 \end{pf}

Therefore an iterative algorithm, using rotations as elementary pivot
permutations, maintains the property that the permutation matrices $P$ and $Q$ at
any step $k$ are $k$-monotonically increasing. A similar property also applies
with recursive algorithms. 
\end{subsubsection}

\subsection{How to reveal rank profiles}
\label{sec:cond}



A PLUQ decomposition reveals the row (resp.\ column) rank profile if it can be
read from the first $r$ values of the permutation matrix $P$ (resp.\ $Q$).
Equivalently, by Lemma~\ref{lem:rpm:prop}, this means that the row (resp.\ column) support of the pivoting
matrix $\Pi_{P,Q}$ equals that of the rank profile matrix.

\begin{defn}
  The decomposition $A=PLUQ$ reveals:
  \begin{compactenum}
  \item the row rank profile if $\RS{\Pi_{P,Q}}= \RS{\RPM{A}}$,
  \item the column rank profile if $\CS{\Pi_{P,Q}}= \CS{\RPM{A}}$,
  \item the rank profile matrix if $\Pi_{P,Q}=\RPM{A}$.
  \end{compactenum}
\end{defn}

\begin{exmp}\label{ex:rankprofile}
$A= 
  \begin{smatrix}
    2 & 0 & 3& 0\\
    1 & 0 & 0& 0\\
    0 & 0 & 4 &0\\
    0 & 2 & 0 &1\\
  \end{smatrix}
  $ 
  has $\RPM{A} = 
  \begin{smatrix}
    1& 0& 0& 0\\
    0& 0& 1& 0 \\
    0& 0& 0& 0\\
    0& 1& 0& 0
  \end{smatrix}
  $ for rank profile matrix over $\mathbb{Q}$.
Now the pivoting matrix obtained from a \pluq decomposition with a pivot search
operation following the row order (any column, first nonzero row) could be the matrix
  $\Pi_{P,Q} = 
\begin{smatrix}
  0&0&1&0\\
  1&0&0&0\\
  0&0&0&0\\
  0&1&0&0\\
\end{smatrix}
$. As these matrices share the same row support, the matrix $\Pi_{P,Q}$ reveals the row
rank profile of $A$.
\end{exmp}

\begin{rem}\label{rem:SwapsConterex}
  Example~\ref{ex:rankprofile} suggests that a pivot search strategy 
  minimizing row and column indices could be a sufficient condition to recover both row
  and column rank profiles at the same time, regardless the pivot permutation.
  However, this is unfortunately not the case. Consider for
  example a search based on  the lexicographic order (first nonzero column of
  the first nonzero row) with transposition permutations, run on the matrix:
  $A= 
  \begin{smatrix}
    0 & 0 & 1\\
    2 & 3 & 0\\
  \end{smatrix}$. Its rank profile matrix is $\RPM{A} = 
  \begin{smatrix}
    0&0&1\\
    1&0&0
  \end{smatrix}
  $ whereas the pivoting matrix would be 
$
  \Pi_{P,Q}=\begin{smatrix}
    0&0&1\\
    0&1&0
  \end{smatrix}
  $, which does not reveal the column rank profile.
 This is due to the fact that the column transposition performed for the
  first pivot changes the order in which the columns will be inspected in
  the search for the second pivot. 
\end{rem}

We will show that if the pivot permutations preserve the order in which the still
unprocessed columns or rows appear, then the pivoting matrix will equal the rank
profile matrix. This is achieved by the monotonically increasing permutations. 
%
%
%
Theorem~\ref{th:RPandPerm} shows how the ability of a \pluq
decomposition algorithm to recover the rank profile information relates to the
use of monotonically increasing permutations.
More precisely, it considers an arbitrary step in a PLUQ decomposition where $k$
pivots have been found in the elimination of an $\ell\times p$ leading sub-matrix $A_1$
of the input matrix~$A$.

\begin{thm}
\label{th:RPandPerm}
  Consider a partial \pluq decomposition of an $m\times n$ matrix $A$:
\[
A = P_1
\begin{bmatrix}
  L_1 \\ M_1 & I_{m-k}
\end{bmatrix}
\begin{bmatrix}
  U_1 & V_1\\
      & H
\end{bmatrix}
Q_1
\] where $
\begin{bmatrix} L_1\\M_1\end{bmatrix}$ is $m\times k$ lower triangular and 
$\begin{bmatrix}  U_1 & V_1\end{bmatrix}$ is $ k\times n$ upper triangular,
and  let $A_1$
be some $\ell \times p$ leading sub-matrix  of $A$, for $\ell,p\geq k$.
Let $H=P_2L_2U_2Q_2$ be a \pluq decomposition of $H$.
Consider the \pluq decomposition
\[
A=\underbrace{P_1
\begin{bmatrix}
  I_k\\& P_2
\end{bmatrix}}_{P}
\underbrace{
\begin{bmatrix}
  L_1\\P_2^TM_1&L_2
\end{bmatrix}}_L
\underbrace{\begin{bmatrix}
  U_1&V_1Q_2^T\\
  &U_2
\end{bmatrix}}_{U}
\underbrace{\begin{bmatrix}
  I_k\\& Q_2
\end{bmatrix}
Q_1}_{Q}.
\]

Consider the following clauses:
  \begin{compactenum}[(i)]
  \item $\RRP(A_1) = \RS{\Pi_{P_1,Q_1}}$ \label{clause:rrp1}
  \item $\CRP(A_1) = \CS{\Pi_{P_1,Q_1}}$ \label{clause:crp1}
  \item $\RPM{A_1} = \Pi_{P_1,Q_1} $ \label{clause:rpm1}
  \item $\RRP(H) = \RS{\Pi_{P_2,Q_2}}$ \label{clause:rrp2}
  \item $\CRP(H) = \CS{\Pi_{P_2,Q_2}} $ \label{clause:crp2}
  \item $\RPM{H} = \Pi_{P_2,Q_2}$ \label{clause:rpm2}
  \end{compactenum}
\begin{compactenum}[(i)]
 \setcounter{enumi}{6}
\item $P_1^T$ is $k$-monotonically increasing or ($P_1^T$ is $\ell$-mono\-tonically
    increasing and $p=n$)\label{clause:PMI} \label{clause:Pmonotinc}
\item $Q_1^T$ is $k$-monotonically increasing or ($Q_1^T$ is $p$-mono\-tonically
    increasing and $\ell=m$)\label{clause:Qmonotinc}
\end{compactenum}
%
%
Then,
\begin{compactenum}[(a)]
\item if (\ref{clause:rrp1}) or (\ref{clause:crp1}) or (\ref{clause:rpm1}) then $H= \begin{bmatrix} 0_{(\ell-k)\times(p-k)}&*\\ *&* \end{bmatrix}$\label{th:H}
\item if (\ref{clause:Pmonotinc})
then ((\ref{clause:rrp1}) and (\ref{clause:rrp2})) $\Rightarrow \RRP(A)=\RS{\Pi_{P,Q}} $;\label{th:RP:row}

\item if (\ref{clause:Qmonotinc})  
then 
((\ref{clause:crp1}) and (\ref{clause:crp2})) $\Rightarrow \CRP(A) = \CS{\Pi_{P,Q}} $;\label{th:RP:col}

\item if (\ref{clause:Pmonotinc}) and (\ref{clause:Qmonotinc})  then (\ref{clause:rpm1}) and (\ref{clause:rpm2}) $\Rightarrow  \RPM{A}=\Pi_{P,Q}$.\label{th:RP:both}
\end{compactenum}

\end{thm}
\begin{pf}
Let $P_1=
\begin{bmatrix}
  P_{11} & E_1
\end{bmatrix}$ and $Q_1=
\begin{bmatrix}
  Q_{11}\\F_1
\end{bmatrix}$
where  $E_1$ is $m\times(m-k)$ and  $F_1$ is $(n-k)\times n$. 
On one hand we have
\begin{eqnarray}
  A &=& 
\underbrace{\begin{bmatrix} P_{11}&E_1\end{bmatrix}
  \begin{bmatrix} L_1\\M_1 \end{bmatrix}
  \begin{bmatrix} U_1&V_1 \end{bmatrix}
  \begin{bmatrix} Q_{11}\\F_1\end{bmatrix}}_{B} + 
    E_1HF_1. \label{eq:RP:PLUQ}
\end{eqnarray}

On the other hand,
  \begin{eqnarray}
    \Pi_{P,Q} &=& P_1 
    \begin{bmatrix}
      I_k\\&P_2
    \end{bmatrix}
    \begin{bmatrix}
      I_r\\&0_{(m-r)\times(n-r)}
    \end{bmatrix}
    \begin{bmatrix}
      I_k\\&Q_2
    \end{bmatrix}Q_1\notag \\
    &=& P_1 
    \begin{bmatrix}
      I_k\\&\Pi_{P_2,Q_2}
    \end{bmatrix}
    Q_1 \notag = \Pi_{P_1,Q_1} + E_1\Pi_{P_2,Q_2}F_1.
    \label{eq:RP:pi}
  \end{eqnarray}

Let $\LA = 
\begin{bmatrix}
  A_1&0\\
  0&0_{(m-\ell)\times(n-p)}
\end{bmatrix}
$ and denote by $B_1$ the $\ell\times p$ leading sub-matrix of~$B$.

\begin{compactenum}[(a)]
\item 
The clause \eqref{clause:rrp1} or \eqref{clause:crp1} or \eqref{clause:rpm1} implies that all $k$ pivots of the partial
  elimination were found within the $\ell\times p$ sub-matrix
  $A_1$. Hence $\rank(A_1)=k$ and we can write 
$P_1=\begin{bmatrix}\begin{array}{c} P_{11}\\0_{(m-\ell)\times k}\end{array}&E_1\end{bmatrix}$
 and $Q_1=\begin{bmatrix} Q_{11}& 0_{k\times(n-p)}\\ \multicolumn{2}{c}{F_1}\end{bmatrix}$,
 and the matrix $A_1$ writes 
 \begin{equation}\label{eq:A1B1}
A_1 = \begin{bmatrix}I_\ell &0 \end{bmatrix}
A
\begin{smatrix}I_p\\0 \end{smatrix} 
= 
B_1 +
 \begin{bmatrix} I_\ell&0 \end{bmatrix}E_1HF_1\begin{smatrix} I_p\\0 \end{smatrix}.
\end{equation}
Now $\rank(B_1)=k$ as a sub-matrix of $B$ of rank $k$ and since
\begin{eqnarray*}
B_1        &=& 
\begin{bmatrix}P_{11}&  \begin{bmatrix}  I_\ell&0 \end{bmatrix}\cdot E_1\end{bmatrix}
\begin{bmatrix} L_1\\M_1\end{bmatrix}
\begin{bmatrix} U_1&V_1\end{bmatrix}
\begin{bmatrix} Q_{11}\\F_1 \cdot \begin{smatrix} I_p\\0 \end{smatrix}\end{bmatrix}  \\
&=&  P_{11}L_1U_1Q_{11} + 
        \begin{bmatrix} I_\ell&0 \end{bmatrix} E_1M_1
        \begin{bmatrix}U_1&V_1\end{bmatrix} Q_1 \begin{smatrix} I_p\\0 \end{smatrix}
  \end{eqnarray*}
where the first term, $P_{11}L_1U_1Q_{11}$, has rank $k$ and the second term has a
disjoint row support.

Finally, consider the term   $\begin{bmatrix}
  I_\ell&0 \end{bmatrix}E_1HF_1\begin{smatrix} I_p\\0 \end{smatrix}$ of
equation~\eqref{eq:A1B1}.  As its row
 support  is disjoint with that of the pivot rows of $B_1$, it has to
be composed of rows linearly dependent with the pivot rows of $B_1$ to ensure
that $\rank(A_1)=k$. As its
column support is disjoint with that of the pivot columns of $B_1$, we conclude
that it must be the zero matrix.
Therefore the
leading $(\ell-k)\times (p-k)$ sub-matrix of $E_1HF_1$ is zero. 
\item From~(\ref{th:H}) we know that $A_1= B_1$. Thus
  $\RRP(B) = \RRP(A_1)$. Recall that $A=B+E_1HF_1$.
No pivot row of $B$ can be made linearly dependent by adding rows of $E_1HF_1$,
as the column position of the pivot is always zero in the latter
matrix. For the same reason, no pivot row of $E_1HF_1$ can be made linearly
dependent by adding rows of $B$.
From~\eqref{clause:rrp1}, the set of pivot rows of $B$ is $\RRP(A_1)$,
which shows that 
\begin{equation}
\RRP(A)=\RRP(A_1)\cup \RRP(E_1HF_1).
\label{eq:rrpa}
\end{equation}

  Let $\sigma_{E_1}:\{1..m-k\}\rightarrow \{1..m\}$ be the map representing the
  sub-permutation $E_1$ (i.e. such that $E_1[\sigma_{E_1}(i),i]=1 \ \forall i$).
  If $P_1^T$ is $k$-monotonically increasing, the matrix $E_1$ has full column
  rank and is in column echelon  form, which implies that 
\begin{eqnarray}
\RRP(E_1HF_1) &=& \sigma_{E_1} (\RRP(HF_1))\notag\\
             &=& \sigma_{E_1}(\RRP(H)),\label{eq:rrpehf}
\end{eqnarray} 
since  $F_1$ has full row rank.
If $P_1^T$ is $\ell$ monotonically increasing, we can write $E_1=
\begin{bmatrix} E_{11}&E_{12}\end{bmatrix}$, where the $m\times (m-\ell)$ matrix
$E_{12}$ is in column echelon form. If $p=n$, the matrix $H$ writes $H=
\begin{bmatrix}  0_{(\ell-k)\times (n-k)}\\H_2 \end{bmatrix}$. Hence we have
$E_1HF_1 = E_{12}H_2F_1$ which also implies
\[
\RRP(E_1HF_1)  = \sigma_{E_1}(\RRP(H)).
\]
  From equation~(\ref{eq:RP:pi}), the row
  support of $\Pi_{P,Q}$ is that of $\Pi_{P_1,Q_1} + E_1\Pi_{P_2,Q_2}F_1$, which is the
  union of the row support of these two terms as they are disjoint. Under the
  conditions of point~\eqref{th:RP:row}, this row support is the union of
  $\RRP(A_1)$ and $\sigma_{E_1}(\RRP(H))$, which is, from~\eqref{eq:rrpehf}
  and~\eqref{eq:rrpa}, $\RRP(A)$.
\item Similarly as for point~(\ref{th:RP:row}).
\item From~(\ref{th:H}) we have still $A_1=B_1$. 
  Now since $\rank(B)=\rank(B_1)=\rank(A_1)=k$, there is no other nonzero element in $\RPM{B}$
  than those in $\RPM{\LA}$ and $\RPM{B}= \RPM{\LA}$.
  The row and column support of  $\RPM{B}$ and that of $E_1HF_1$ are disjoint. Hence 
\begin{equation}\label{eq:RP:RA}
\RPM{A} =    \RPM{\LA} + \RPM{E_1HF_1}.
\end{equation}
If both $P_1^T$ and $Q_1^T$ are $k$-monotonically increasing, the matrix $E_1$ is in
  column echelon form and the matrix $F_1$ in row echelon form. Consequently, the matrix
  $E_1HF_1$ is a copy of the matrix $H$ with $k$ zero-rows and $k$ zero-columns
  interleaved, which does not impact the linear dependency relations between 
  the nonzero rows and columns. As a consequence 
\begin{equation}
\RPM{E_1HF_1} =  E_1\RPM{H}F_1. \label{eq:EHFRPM}
\end{equation}
 Now if $Q_1^T$ is $k$-monotonically increasing, $P_1^T$ is
  $\ell$-mono\-ton\-ically increasing and $p=n$, then, using notations of
  point~\eqref{th:RP:row}, $E_1HF_1 = E_{12}H_2F_1$ where
  $E_{12}$ is in column echelon form. Thus  $\RPM{E_1HF_1} =  E_1\RPM{H}F_1$ for
  the same reason. The symmetric case where $Q_1^T$ is $p$-monotonically
  increasing and $\ell=m$ works similarly.
  Combining equations~(\ref{eq:RP:pi}),~(\ref{eq:RP:RA}) and~(\ref{eq:EHFRPM})  gives 
  $\RPM{A}  = \Pi_{P,Q}$.%
\end{compactenum}%
\end{pf}

\section{Algorithms for the rank profile matrix}\label{sec:algo}
Using Theorem~\ref{th:RPandPerm}, we  deduce what rank profile information is
revealed by a PLUQ algorithm by the way the search and the permutation
operations are done.
Table~\ref{tab:RPRSearchPerm} summarizes these results, and points to instances
known in the literature, implementing the corresponding type of elimination.
More precisely, we first distinguish in this table the ability to compute the
row or column rank profile or the rank profile matrix, but we also indicate
whether the resulting PLUQ decomposition preserves the monotonicity of the rows
or columns. Indeed some algorithm may compute the rank profile matrix,
but break the precedence relation between the linearly dependent rows or
columns,  making it unusable as a base case for a block algorithm of higher level.

\begin{table*}[htb]\small
\begin{center}
\resizebox{\linewidth}{!}{%
\begin{tabular}{llllll}
\toprule
\textbf{Search }&\textbf{Row
  Perm.}&\textbf{Col. Perm.}&\textbf{Reveals}&\textbf{Monotonicity} & \textbf{Instance}\\
\midrule
 Row order  & Transposition & Transposition &\RRP & &\cite{IMH:1982,JPS:2013}\\ 
 Col. order  & Transposition & Transposition &\CRP & & \cite{KG:1985,JPS:2013}\\ 
\midrule
  \multirow{3}{*}{Lexico.} & Transposition & Transposition & \RRP & &\cite{Storjohann:2000:thesis}\\
  & Transposition &Rotation & \RRP, \CRP, \RPM{}&Col. &here \\
  & Rotation &Rotation & \RRP, \CRP, \RPM{}&Row, Col. &here \\
\midrule
\multirow{3}{*}{Rev. lexico.} & Transposition & Transposition & \CRP & &\cite{Storjohann:2000:thesis}\\
   & Rotation & Transposition & \RRP, \CRP, \RPM{}& Row& here\\
   & Rotation & Rotation & \RRP, \CRP, \RPM{}& Row, Col.& here\\
\midrule
  \multirow{3}{*}{Product}  & Rotation &Transposition &\RRP& Row &here \\
    & Transposition &Rotation &\CRP& Col &here \\
    & Rotation &Rotation &\RRP, \CRP, \RPM{}& Row, Col.&\cite{DPS:2013} \\
\bottomrule
\end{tabular}
}
\caption{Pivoting Strategies revealing rank profiles}\label{tab:RPRSearchPerm}
\end{center}
\end{table*}

\subsection{Iterative algorithms}
We start with iterative algorithms, where each iteration handles one pivot at a
time. Here Theorem~\ref{th:RPandPerm} is applied with $k=1$, and
the partial elimination represents how one pivot is being treated. The
elimination of $H$ is done by induction. 

\subsubsection{Row and Column order Search}

The row order pivot search operation is of the form: 
\textit{any nonzero element in the first nonzero row}. 
Each row is inspected in order, and a new row is considered only
when the previous row is all zeros.  
With the notations of Theorem~\ref{th:RPandPerm}, this means that $A_1$ is the
leading $\ell\times n$ sub-matrix of $A$, where $\ell$ is the index of 
the first nonzero row of $A$.
When permutations $P_1$ and $Q_1$, moving the pivot from
position $(\ell,j)$ to $(k,k)$ are transpositions, 
the matrix $\Pi_{P_1,Q_1}$ is the element $E_{\ell,j}$ of the canonical basis. 
Its row rank profile is $(\ell)$ which is that of the $\ell
\times n$ leading sub-matrix $A_1$. Finally, the permutation $P_1$ is
$\ell$-monotonically increasing, and Theorem~\ref{th:RPandPerm}
case~(\ref{th:RP:row}) can be applied to prove by induction that any such
algorithm will reveal the row rank profile: $\RRP(A)=\RS{\Pi_{P,Q}}$.
The case of the column order search is similar. 
%
\subsubsection{Lexicographic order based pivot search}

In this case the pivot search operation is of the form: 
\textit{first nonzero element in the first nonzero row}. 
The lexicographic order being compatible with the row order, the above results
hold when transpositions are used and the row rank profile is revealed. If in
addition column rotations are used, $Q_1=R_{1,j}$ which is $1$-monotonically
increasing. Now $\Pi_{P_1,Q_1}=E_{\ell,j}$ which is the rank profile matrix of
the $\ell\times n$ leading sub-matrix $A_1$ of $A$. Theorem~\ref{th:RPandPerm}
case~(\ref{th:RP:both}) can be applied to prove by induction that any such
algorithm will reveal the rank profile matrix: $\RPM{A}=\Pi_{P,Q}$. 
Lastly, the use of row rotations, ensures that the order of the linearly
dependent rows will be preserved as well.
Algorithm~\ref{alg:pluq:iter} is an instance of Gaussian elimination with a
lexicographic order search and rotations for row and column permutations.

The case of the reverse lexicographic order search is similar.
As an example, the algorithm in~\cite[Algorithm 2.14]{Storjohann:2000:thesis} is based
on a reverse lexicographic order search but with transpositions for the row
permutations. Hence it only reveals the column rank profile.

\subsubsection{Product order based pivot search}

The search here consists in finding any nonzero element $A_{\ell,p}$ such that
the $\ell\times p$ leading sub-matrix $A_1$ of $A$ is all zeros except this
coefficient. If the row and column permutations are the rotations $R_{1,\ell}$
and $R_{1,p}$, we have
$\Pi_{P_1,Q_1}=E_{\ell,p}=\RPM{A_1}$. Theorem~\ref{th:RPandPerm}
case~(\ref{th:RP:both}) can be applied to prove by induction that any such
algorithm will reveal the rank profile matrix: $\RPM{A}=\Pi_{P,Q}$. 
An instance  of such an algorithm is given in~\cite[Algorithm~2]{DPS:2013}.
If $P_1$ (resp.\ $Q_1$) is a transposition, then 
Theorem~\ref{th:RPandPerm} case~(\ref{th:RP:col}) (resp.\ case~(\ref{th:RP:row}))
applies to show by induction that the columns (resp.\ row) rank profile is revealed.

\subsection{Recursive algorithms}


A recursive Gaussian elimination algorithm can either split one of the row or
column dimension, cutting the matrix in  wide or tall rectangular slabs, or
split both dimensions, leading to a decomposition into tiles.

\subsubsection{Slab recursive algorithms}
Most algorithms computing rank profiles are slab
recursive~\cite{IMH:1982,KG:1985,Storjohann:2000:thesis,JPS:2013}.
When the row dimension is split, this means that the search space for pivots is
the whole set of columns, and Theorem~\ref{th:RPandPerm} applies with
$p=n$. This corresponds to either a row or a lexicographic order.
From case(~\ref{th:RP:row}), one shows that, with transpositions, the algorithm
recovers the row rank profile, provided that the base case does.
If in addition, the elementary column permutations are rotations, then
case~(\ref{th:RP:both}) applies and the rank profile matrix is recovered.
Finally, if rows are also permuted by monotonically increasing permutations,
then the PLUQ decomposition also respects the monotonicity of the linearly
dependent rows and columns.
The same reasoning holds when splitting the column dimension.

\subsubsection{Tile recursive algorithms}

Tile recursive Gaussian elimination
algorithms~\cite{DPS:2013,Malaschonok:2010,DuRo:2002} are more involved, 
especially when dealing with rank deficiencies. Algorithm~\ref{alg:pluq:4rec}
recalls the tile recursive algorithm that the authors proposed
in~\cite{DPS:2013}.

{
\begin{algorithm}[htbp]
  \caption{\pluq}\label{alg:pluq:4rec}
\begin{algorithmic}
\Require{$A=(a_{ij})$ a $m\times n$ matrix over a field}
\Ensure{$P, Q$:  $m\times m$ and $n\times n$ permutation matrices, $r$, the rank of $A$}
\Ensure{$A \leftarrow
\begin{smatrix}
    L \backslash U & V\\
    M              &0
  \end{smatrix}$ where 
$L$ and $U$ are $r\times r$ resp. unit lower and upper triangular, 
and $A= P \begin{smatrix} L\\M \end{smatrix} \begin{smatrix} U&V \end{smatrix} Q.$
}
\If{$\min(m,n)\leq \text{Threshold}$}
   \State Call a base case implementation
\EndIf
\State Split $A=
    \begin{smatrix}
      A_{1}&A_{2}\\
      A_{3}&A_{4}\\
    \end{smatrix}$ where $A_{1}$ is $\lfloor\frac{m}{2}\rfloor\times \lfloor\frac{n}{2}\rfloor$.
  \State Decompose $A_{1} = P_1
  \begin{smatrix}L_1\\M_1\end{smatrix}\begin{smatrix}U_1&V_1\end{smatrix}
  Q_1$ \Comment{$\pluq(A_{1})$}
  \State $\begin{smatrix} B_1\\B_2\end{smatrix}
\leftarrow P_1^TA_{2}$ \Comment{$\texttt{PermR}(A_2,P_1^T)$}
  \State $
  \begin{smatrix}
    C_1&C_2
  \end{smatrix}
 \leftarrow A_{3}Q_1^T$ \Comment{$\texttt{PermC}(A_3,Q_1^T)$}
  \State Here $A =
\begin{smatrix}
      L_1 \backslash U_1& V_1& B_1\\
      M_1               & 0  & B_2\\
      C_1               & C_2& A_{4}\\
\end{smatrix}
$.
   \State $D\leftarrow L_1^{-1}B_1$ \Comment{$\trsm(L_1,B_1)$}
   \State $E\leftarrow C_1U_1^{-1}$ \Comment{$\trsm(C_1,U_1)$}
   \State $F\leftarrow B_2-M_1D$ \Comment{$\MM(B_2,M_1,D)$}
   \State $G\leftarrow C_2-EV_1$ \Comment{$\MM(C_2,E,V_1)$}
   \State $H\leftarrow A_4-ED$ \Comment{$\MM(A_4,E,D)$}
   \State Here $A=
   \begin{smatrix}
     L_1 \backslash U_1& V_1& D\\
       M_1               & 0  & F\\
       \hline
       E               & G& H\\
   \end{smatrix}
$.
   
   \State Decompose $F =
   P_2\begin{smatrix}L_2\\M_2\end{smatrix}\begin{smatrix}U_2&V_2\end{smatrix}
   Q_2$ \Comment{$\pluq(F)$}
   \State Decompose $G =
   P_3\begin{smatrix}L_3\\M_3\end{smatrix}\begin{smatrix}U_3&V_3\end{smatrix}
   Q_3$ \Comment{$\pluq(G)$}
   \State $ \begin{smatrix} H_1&H_2\\H_3&H_4 \end{smatrix}
\leftarrow P_3^THQ_2^T$ \Comment{$\texttt{PermR}(H,P_3^T);\texttt{PermC}(H,Q_2^T)$}
   \State $\begin{smatrix} E_1\\E_2 \end{smatrix} \leftarrow  P_3^T E$ \Comment{$\texttt{PermR}(E,P_3^T)$}
   \State $\begin{smatrix} M_{11}\\M_{12} \end{smatrix} \leftarrow  P_2^T M_1$ \Comment{$\texttt{PermR}(M_1,P_2^T)$}
   \State $\begin{smatrix} D_1&D_2\end{smatrix} \leftarrow  DQ_2^T$ \Comment{$\texttt{PermR}(D,Q_2^T)$}
   \State $\begin{smatrix} V_{11}&V_{12}\end{smatrix} \leftarrow  V_1Q_3^T$ \Comment{$\texttt{PermR}(V_1,Q_3^T)$}

\State Here $A=
\begin{smatrix}
      L_1 \backslash U_1& V_{11}&V_{12}& D_1 &D_2\\
      M_{11}               & 0  &0& L_2\backslash U_2 & V_2\\
      M_{12}               & 0  &0& M_2 & 0\\
      E_1               & L_3\backslash U_3&V_3& H_1 &H_2\\
      E_2               & M_3              &0  & H_3 &H_4\\
\end{smatrix}
$.
   \State $I\leftarrow H_1U_2^{-1}$ \Comment{$\trsm(H_1,U_2)$}
   \State $J\leftarrow L_3^{-1}I$ \Comment{$\trsm(L_3,I)$}
    \State $K\leftarrow H_3U_2^{-1}$ \Comment{$\trsm(H_3,U_2)$}
    \State $N\leftarrow L_3^{-1}H_2$ \Comment{$\trsm(L_3,H_2)$}
    \State    $O\leftarrow N-JV_2$ \Comment{$\MM(N,J,V_2)$}
    \State $R\leftarrow H_4-KV_2-M_3O$ \Comment{$\MM(H_4,K,V_2);\MM(H_4,M_3,O)$}
    \State  Decompose $ R =
    P_4\begin{smatrix}L_4\\M_4\end{smatrix}\begin{smatrix}U_4&V_4\end{smatrix}
    Q_4$ \Comment{$\pluq(R)$}
    \State $ \begin{smatrix}
      E_{21}& M_{31} & 0 & K_1\\
      E_{22}& M_{32} & 0 & K_2\\
    \end{smatrix} 
    \leftarrow  P_4^T  \begin{smatrix} E_{2}& M_{3} & 0 & K\\ \end{smatrix}$ \Comment{$\texttt{PermR}$}
    \State $\begin{smatrix} D_{21}&D_{22}\\ V_{21}&V_{22}\\ 0&0 \\ O_1&O_2 \end{smatrix}
            \leftarrow\begin{smatrix} D_{2}\\ V_{2}\\ 0 \\ O\end{smatrix} Q_4^T$ \Comment{$\texttt{PermC}$}
\breakalgorithm{}
    \State Here $A=
    \begin{smatrix}
      L_1 \backslash U_1& V_{11}&V_{12}& D_1 &D_{21}&D_{22}\\
      M_{11}               & 0  &0& L_2\backslash U_2 & V_{21}&V_{22}\\
      M_{12}               & 0  &0& M_2 & 0&0\\
      E_1               & L_3\backslash U_3&V_3& I &O_1 &O_2\\
      E_{21}               & M_{31}              &0  & K_1 & L_4\backslash U_4 &V_4\\
      E_{22}               & M_{32}              &0  & K_2 & M_4 &0\\
    \end{smatrix}
$.
\State $S\leftarrow \begin{smatrix}
  I_{r_1+r_2}\\
  &&I_{k-r_1-r_2}\\
  &I_{r_3+r_4}\\
  &&&&I_{m-k-r_3-r_4}
\end{smatrix}$
\State $T\leftarrow \begin{smatrix}
  I_{r_1}\\
        &      &          &I_{r_2} & \\
        &I_{r_3}&          &       & \\
        &      &          &       & I_{r_4}\\
        &      &I_{k-r_1-r_3}\\
        &      &          &       &       & I_{n-k-r_2-r_4}\\
\end{smatrix}$

 \State $P\leftarrow \text{Diag}(
   P_1
   \begin{smatrix}
     I_{r_1}\\&P_2
   \end{smatrix},P_3
   \begin{smatrix}
     I_{r_3}\\&P_4
   \end{smatrix} ) S$
 \State $Q\leftarrow
T\text{Diag}(\begin{smatrix}
     I_{r_1}\\&Q_3
   \end{smatrix}
 Q_1,\begin{smatrix}
     I_{r_2}\\&Q_4
   \end{smatrix}
 Q_2)$

 
\State $A\leftarrow S^TAT^T$ \Comment{$\texttt{PermR}(A,S^T); \texttt{PermC}(A,T^T)$}
\State Here $A=
\begin{smatrix}
  L_1\backslash U_1 &D_1              &V_{11} &D_{21} & V_{12}&D_{22}\\
    M_{11}          &L_2\backslash U_2 & 0 &       V_{21} & 0 & V_{22}\\ 
     E_1           &       I         &L_3\backslash U_3 & O_1 & V_3&O_2 \\
     E_{21}         &K_1              & M_{31}            & L_4\backslash U_4     &0&V_4 \\
     M_{12}         & M_2             & 0 &0 &0& 0\\
     E_{22}         & K_2             & M_{32}           & M_4 & 0 &0\\
\end{smatrix}$
\Return $(P,Q,r_1+r_2+r_3+r_4,A)$
\end{algorithmic}
\end{algorithm}
}
Here, the search area $A_1$ has arbitrary dimensions $\ell\times p$, often
specialized as $m/2 \times n/2$. As a consequence, the pivot search can not
satisfy either a row, a column, a lexicographic or a reverse lexicographic order.
Now, if the pivots selected in the elimination of  $A_1$ minimize the product
order, then they necessarily also respect this order as pivots of the whole
matrix $A$. Now, from~(\ref{th:H}), the remaining matrix $H$ writes
$H=\begin{smatrix}  0_{(\ell-k)\times(p-k)} & H_{12}\\H_{21} &
  H_{22}\end{smatrix}$ and its elimination is done by two independent eliminations
on the blocks $H_{12}$ and $H_{21}$, followed by some update of $H_{22}$ and a
last elimination on it. Here again, pivots minimizing the row order on $H_{21}$
and $H_{12}$ are also pivots minimizing this order for $H$, and so are those of
the fourth elimination. Now the block row and column permutations used
in~\cite[Algorithm~1]{DPS:2013} to form the PLUQ decomposition are
$r$-monotonically increasing. Hence, from case~(\ref{th:RP:both}), the algorithm
computes the rank profile matrix and preserves the monotonicity.
If only one of the row or column permutations are rotations, then
case~(\ref{th:RP:row}) or~(\ref{th:RP:col}) applies to show that either the row or
the column rank profile is computed.


\section{Improvements in practice}
\label{sec:basecase}
In~\cite{DPS:2013}, we identified the ability to
recover the rank profile matrix via the use of the product order search and
of rotations.
Hence we proposed an implementation combining a tile recursive algorithm and an iterative
base case, using these search and permutation strategies.

The analysis of sections~\ref{sec:cond} and~\ref{sec:algo} shows that other
pivoting strategies can be used to compute the rank profile matrix, and preserve
the monotonicity. We present here a new base case algorithm and its
implementation over a finite field that we wrote in the \fflasffpack
library\footnote{FFLAS-FFPACK Git rev. d420b4b,
  \url{http://linbox-team.github.io/fflas-ffpack/}, linked against OpenBLAS-v0.2.9.}. 
It is based on a
lexicographic order search and row
and column rotations. Moreover, the schedule of the update operations is that
of a Crout elimination, for it reduces the number of modular reductions, as
shown in~\cite[\S~3.1]{DPS:2014}.
Algorithm~\ref{alg:pluq:iter} summarizes this variant.
\begin{algorithm}[htbp]
  \caption{Crout variant of \pluq with lexicographic search and column rotations}
  \label{alg:pluq:iter}
\begin{algorithmic}[1]
\State $k\leftarrow 1$
\For{$i=1\dots m$}
   \State $A_{i,k..n}\leftarrow A_{i,k..n} - A_{i,1..k-1}\times A_{1..k-1,k..n}$
   \If{$A_{i,k..n} = 0$}
     \State Loop to next iteration
   \EndIf
   \State Let $A_{i,s}$ be the left-most nonzero element of row $i$.
   \State $A_{i+1..m,s}\leftarrow A_{i+1..m,s} - A_{i+1..m,1..k-1}\times A_{1..k-1,s}$
   \State $A_{i+1..m,s}\leftarrow A_{i+1..m,s} / A_{i,s}$
   \State Bring $A_{*,s}$ to $A_{*,k}$ by column rotation
   \State Bring $A_{i,*}$ to $A_{k,*}$ by row rotation
   \State $k\leftarrow k+1$
\EndFor
\end{algorithmic}
\end{algorithm}

\begin{figure}[ht]
  \centering
  \includegraphics[width=.8\columnwidth]{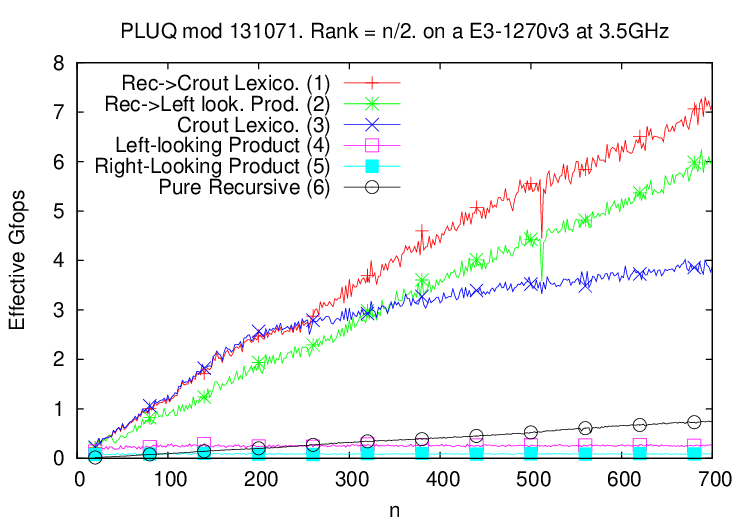}\\
  \includegraphics[width=.8\columnwidth]{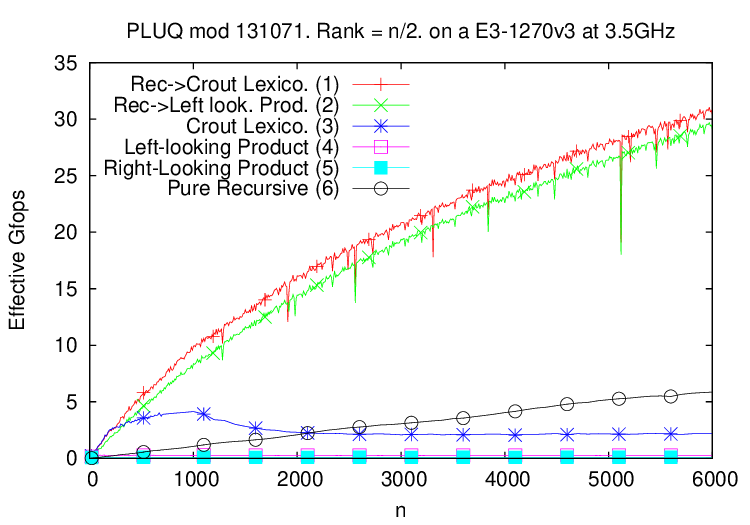} 
  \caption{Computation speed of PLUQ decomposition base cases.}
  \label{fig:basecase}
\end{figure}
In the following experiments, we report the median of the real time for up to 60
computations (depending on the size of the instance), with $n\times n$ matrices
of rank $r=n/2$ (left) or $r=n/8$ (right).
In order to ensure that the row and column rank profiles of these
matrices are random, we construct them as the product $A=L\RPM{}U$, where
$L$ and $U$ are random nonsingular lower and upper triangular matrices and $\RPM{}$ is
an $m\times n$ $r$-sub-permutation matrix whose nonzero elements positions are chosen
uniformly at random.
The effective speed is obtained by dividing an estimate of the arithmetic cost
($2mnr+2/3r^3-r^2(m+n)$) by the computation time. 

Figure~\ref{fig:basecase} shows the computation speed of Algorithm~\ref{alg:pluq:iter} (3),
compared to that of the pure recursive algorithm (6), and to our
previous base case algorithm in~\cite{DPS:2013}, using a product order search, and either a
left-looking (4) or a right-looking (5) schedule. At $n=200$, the left-looking
variant (4) improves over the right looking variant (5) by a factor of about
$3.18$ as it performs fewer modular reductions. Then, the Crout variant (3)
again improves variant (4) by a factor of about $9.29$. Lastly we also show the speed
of the final implementation, formed by the tile recursive algorithm cascading to either the
Crout base case (1) or the left-looking one (2). The threshold where the
cascading to the base case occurs is experimentally set to its optimum value,
i.e. 256 for variant (1) and 40 for variant (2). This illustrates that the
gain on the base case efficiency leads to a higher threshold, and
 improves the efficiency of the cascade implementation (by an additive
gain of about 1 effective Gfops in the range of dimensions considered).

Note that the experiments reported here differ from that
of~\cite{DPS:ELU:2015}. We found a mistake in the code generating the random
matrices, making their row rank profile generic, which led to a reduced
amount of row permutations and therefore a much higher computation speed. After
fixing this issue, we noticed a slow-down of most base case implementations, but
all variants still compare in the same way. The spikes in curves (1) and (2)
showing drops in computation speed in a few points are reproducible, and seem to
correspond to pathological dimensions, for which the blocking generates worst
case cache misses scenari.

\section{Relations with other triangularizations}
\label{sec:relations}
\subsection{The LEU decomposition}\label{ssec:LEU}

The LEU decomposition of~\cite{Malaschonok:2010} involves a lower
triangular matrix $L$, an upper triangular matrix $U$ and a $r$-sub-permutation
matrix $E$. It is also called the {\em modified Bruhat decomposition}
in~\cite{Tyrtyshnikov:1997:Bruhat}.
Theorem~\ref{thm:leu} shows how to recover an LEU decomposition from
a PLUQ decomposition revealing the rank profile matrix.

\begin{thm}\label{thm:leu}
  Let $A=PLUQ$ be a PLUQ decomposition revealing the rank profile matrix
  ($\Pi_{P,Q}=\RPM{A}$). Then an LEU decomposition of $A$ with $E=\RPM{A}$
  is obtained as follows (only using row and column permutations):
\begin{equation}\label{eq:LEU}
 A = \underbrace{P\begin{bmatrix} L&0_{m\times (m-r)}\end{bmatrix}P^T}_{\overline{L}} 
 \underbrace{P\begin{bmatrix}I_r\\&0\end{bmatrix}Q}_{E} 
 \underbrace{Q^T\begin{bmatrix} U\\0_{(n-r)\times n}\end{bmatrix}Q}_{\overline{U}}
\end{equation}
 \end{thm}
 \begin{pf}
First $E=P\begin{smatrix}I_r\\&0\end{smatrix}Q = \Pi_{P,Q} = \RPM{A}$.
Then there only needs to show that $\overline{L}$ is lower triangular and
   $\overline{U}$ is upper triangular.
Suppose that $\overline{L}$ is not lower triangular, let $i$ be the first row
index such that $\overline{L}_{i,j}\neq0$ for some $i<j$. 
First $j\in\RRP(A)$ since the nonzero columns in $\overline{L}$ are placed
according to the first $r$ values of $P$.
Remarking that $A = P \begin{smatrix} L&0_{m\times (n-r)}\end{smatrix}
\begin{smatrix}  &{U}\\ 0 & I_{n-r} \end{smatrix} Q$, and since
right multiplication by a nonsingular matrix does not change row rank profiles,
we deduce that
$\RRP(\Pi_{P,Q})=\RRP(A)=\RRP(\overline{L})$.
If $i\notin\RRP(A)$, then the $i$-th row of $\overline{L}$ is linearly
dependent with the previous rows, but none of them has a nonzero element in
column $j>i$. Hence $i\in\RRP(A)$.

Let $(a,b)$ be the position of the coefficient $\overline{L}_{i,j}$ in $L$, that
is $a=\sigma_P^{-1}(i), b=\sigma_P^{-1}(j)$. Let also $s=\sigma_Q(a)$ and
$t=\sigma_Q(b)$ so that the pivots at diagonal position $a$ and $b$ in $L$
respectively correspond to ones in $\RPM{A}$ at positions $(i,s)$ and $(j,t)$.
Consider the $\ell\times p$ leading sub-matrices $A_1$ of $A$ where
$\ell=\max_{x=1..a-1}(\sigma_P(x))$ and
$p=\max_{x=1..a-1}(\sigma_Q(x))$.
On one hand $(j,t)$ is an index position in $A_1$ but not $(i,s)$, since
otherwise $\rank(A_1) = b$.
Therefore, $(i,s) \nprec_{prod} (j,t)$, and $s>t$ as $i<j$.
As coefficients $(j,t)$ and $(i,s)$ are pivots in $\RPM{A}$ and $i<j$ and $t<s$,
there can not be a nonzero element above $(j,t)$ at row $i$ when it is chosen
as a pivot. Hence $\overline{L}_{i,j}=0$ and $\overline{L}$ is lower triangular.
The same reasoning applies to show that $\overline{U}$ is upper triangular.
 \end{pf}

\begin{rem}\label{rem:cexleuuniq} 
Note that the LEU decomposition with $E=\RPM{A}$ is not unique, even for
invertible matrices. As a counter-example, the following decomposition holds for
any $a\in\K$:
\begin{equation}
\left[\begin{matrix} 0 & 1 \\ 1 & 0 \end{matrix}\right]
=
\left[\begin{matrix} 1 & 0 \\ a & 1 \end{matrix}\right]
\left[\begin{matrix} 0 & 1 \\ 1 & 0 \end{matrix}\right]
\left[\begin{matrix} 1 & -a \\ 0 & 1 \end{matrix}\right]
\end{equation}

\end{rem}
However, uniqueness of the full rank LEU decomposition can be achieved by imposing an
additional constraint on the $L$ matrix: there is a unique LEU decomposition
where $L$ is such that $E^T L E$ is also lower triangular~\cite[Th. 2.1]{Manthey:2007:Bruhat}.
We remark that for the LEU decomposition proposed in~\eqref{eq:LEU}, this
condition becomes: 
\begin{equation}
  \label{eq:uniqueLEU}
\text{the matrix}~E^T \bar L E = Q_{*,1..r}^T L_{1..r,1..r} Q_{*,1..r}~\text{must be lower triangular.}
\end{equation}
Whenever $Q_{*,1..r}$ is in column echelon form, which means that the sequence
of the pivot's column positions is monotonically increasing, then the above
matrix is lower triangular. A sufficient condition to ensure it, is
thus that the pivoting strategy of the corresponding PLUQ
decomposition  minimizes the reverse lexicographic order.

\begin{cor}
  Let $A=PLUQ$ be a full-rank PLUQ decomposition computed with a pivoting strategy
  minimizing the reverse lexicographic order and performing row rotations. Then
  the LEU decomposition of equation~\eqref{eq:LEU} is the unique
  modified Bruhat decomposition as defined
  in~\cite[Theorem~2.1]{Manthey:2007:Bruhat}. 
\end{cor}

\subsection{The Bruhat decomposition}
The Bruhat decomposition, from which  Malascho\-nok's LEU
decomposition was inspired in~\cite{Malaschonok:2010}, is another decomposition with a
central permutation 
matrix, see~\cite{Bruhat:1956:Lie,delladora:1973:these,Bourbaki:2008:lie,Grigoriev:1981:bruhat}.
Bruhat's theorem is stated in terms of Weyl groups but reduces to the following:
\begin{thm}[\cite{Bourbaki:2008:lie}]
Any invertible matrix $A$ can be written as $A=V P U$ for $V$ and $U$ uppper triangular invertible matrices and $P$ a permutation matrix. The latter decomposition is called the {\em Bruhat decomposition} of $A$.
\end{thm}
It was then naturally extended to singular square matrices
in~\cite{Grigoriev:1981:bruhat} and an algorithm over principal ideals
domains given
in~\cite{Malaschonok:2013:CASC}.  
Corollary~\ref{cor:bruhat} generalizes it to matrices with arbitrary dimensions,
and relates it to the PLUQ decomposition.
\begin{cor}\label{cor:bruhat}
  Any $m\times n$ matrix of rank $r$ has a $VPU$ decomposition, where $V$ and $U$ are upper
  triangular matrices, and $P$ is a $r$-sub-permutation matrix.
\end{cor}
\begin{pf}
Let $J_n$ be the unit anti-diagonal matrix. From the LEU decomposition of $J_nA$,
we have
  $A= \underbrace{J_nLJ_n}_V\underbrace{J_nE}_P U$ where $V$ is upper triangular.
\end{pf}
\begin{rem} \cite{Manthey:2007:Bruhat} gives also a unique {\em generalized Bruhat
    decomposition}, $A=XPY$, but with non-square and non-triangular matrices $X$,
  $Y$. There, $P$ is called the {\em Bruhat permutation} but contains only the
  $r$ nonzero rows and columns of the rank profile matrix. We show in
  Section~\ref{sec:XPY}, that $X$ and $Y$ are actual row and column echelon
  forms of the matrix.
\end{rem}

\subsection{Relation to LUP and PLU decompositions}

The \lup decomposition $A=LUP$ only exists for matrices with generic row rank profile
(including matrices with full row rank). Corollary~\ref{cor:elu} shows upon
which condition the permutation matrix $P$ equals the rank profile matrix $\RPM{A}$.
Note that although the row rank profile of $A$ is trivial in such
cases, the matrix $\RPM{A}$ still carries non-trivial information not only on the column rank
profile but also on the row and column rank profiles of any leading sub-matrix of $A$.
 
\begin{cor}\label{cor:elu}
  Let $A$ be an $m\times n$ matrix.

If $A$ has generic column rank profile, then any \plu decomposition $A=PLU$
    computed using reverse lexicographic order search and row rotations is such
    that 
    $\RPM{A}= P \begin{smatrix}I_r\\&0\end{smatrix}$. In particular,
    $P=\RPM{A}$ if $r=m=n$.
  
  If $A$ has generic row rank profile, then any \lup decomposition $A=LUP$
    computed using lexicographic order search and column rotations is such that 
    $\RPM{A}= \begin{smatrix}I_r \\ & 0\end{smatrix}P$. In particular,
    $P=\RPM{A}$ if $r=m=n$.
\end{cor}

\begin{pf}
  Consider $A$ has generic column rank profile.
  From table~\ref{tab:RPRSearchPerm}, any \pluq decomposition algorithm with a reverse
  lexicographic order based search and rotation based row permutation is such
  that $\Pi_{P,Q}=P \begin{smatrix}  I_r\\& \end{smatrix}Q = \RPM{A}$. Since the
  search follows the reverse lexicographic order and the matrix has generic
  column rank profile, no column will be permuted in this elimination, and
  therefore $Q=I_n$.
  The same reasoning hold for when $A$ has generic row rank profile.
\end{pf}

Note that the $L$ and $U$ factors in a \plu decomposition are uniquely determined
by the permutation $P$. Hence, when the matrix has full row rank, $P=\RPM{A}$
and the decomposition $A=\RPM{A}LU$ is unique. Similarly the decomposition
$A=LU\RPM{A}$ is unique when the matrix has 
full column rank.
Now when the matrix is rank deficient with generic row rank profile, there is no
longer a unique \plu decomposition revealing the rank profile matrix: any
permutation applied to the last $m-r$ columns of $P$ and the last $m-r$ rows of
$L$ yields a \plu decomposition where $\RPM{A}=P \begin{smatrix} I_r\\& \end{smatrix}$. 

Lastly, we remark that the only situation where the rank profile matrix \RPM{A}
can be read directly as a sub-matrix of $P$ or $Q$ is as in
corollary~\ref{cor:elu}, when the matrix $A$ has generic row or column rank
profile. 
Consider a \pluq decomposition $A=PLUQ$ revealing the rank profile matrix
($\RPM{A}=P\begin{smatrix} I_r\\&\end{smatrix}Q$) such that $\RPM{A}$ is a
sub-matrix of $P$. This means that $P=\RPM{A}+S$ where $S$ has disjoint row and
column support with $\RPM{A}$.
We have
$  \RPM{A} = (\RPM{A}+S)  \begin{smatrix}    I_r\\&  \end{smatrix}  Q
          = (\RPM{A}+S)  \begin{smatrix}    Q_1\\0_{(n-r)\times
      n}  \end{smatrix}  
$.
Hence $\RPM{A}(I_n- \begin{smatrix}  Q_1\\0_{(n-r)\times n}  \end{smatrix}) =S \begin{smatrix}  Q_1\\0_{(n-r)\times n}  \end{smatrix}
$ but the row support of these matrices are disjoint, hence 
$\RPM{A}\begin{smatrix}  0\\I_{n-r}\end{smatrix}=0$ which implies that $A$ has
generic column rank profile.
Similarly, one shows that $\RPM{A}$ can be a sub-matrix of $Q$ only if $A$ has a
generic row rank profile.

\subsection{Computing Echelon forms}\label{sec:echelon}

Usual algorithms computing an echelon
form by~\cite{IMH:1982,KG:1985,Storjohann:2000:thesis,JPS:2013} use a slab block decomposition (with 
row or lexicographic order search), which implies that pivots appear in the
order of the echelon form. The column echelon form is simply obtained as $C=PL$
from the PLUQ decomposition.

We extend here this approach to any PLUQ decomposition that reveals the rank
profile matrix ($\RPM{A}=\Pi_{P,Q}$) and show how both the row and column
echelon forms can be recovered by only permutations on the $L$ and $U$ matrices.
Besides the ability to compute both echelon forms from one Gaussian elimination,
this enables one to choose from a broader range of algorithmic variants for
this Gaussian elimination.

Consider a PLUQ decomposition $A=PLUQ$ revealing the rank profile matrix and let 
$P_C = P_{*,1..r}$ be the first $r$ columns of $P$. The ones in $P_C$ indicate
the row position of the pivots in the column echelon form of $A$, but they may not
appear in the same column order. There exist an $r\times r$ permutation matrix
$S$ sorting the columns of $P_C$, i.e. such that $P_S = P_C S$ is in column echelon.

\begin{lem}\label{lem:echelon}
  The matrix $\begin{bmatrix} PL  S&0_{m\times n-r}  \end{bmatrix} $ is a column echelon form of $A$.
\end{lem}

\begin{proof}

$PLS = P \begin{bmatrix}  L&0_{m\times (m-r)}\end{bmatrix} P^T P
\begin{smatrix}   S\\0_{(m-r)\times r}\end{smatrix} = \bar L P_S$.
From Theorem~\ref{thm:leu}, the matrix $\bar L = P \begin{bmatrix}  L & 0_{m\times (m-r)}\end{bmatrix} P^T$
is lower triangular.
Multiplying $\bar L$ on the right by $P_S$ simply removes the zero columns in
$\bar L$ putting it in column echelon form.
Now the relation
$
A=\begin{bmatrix}  PLS & 0_{m\times(n-r)}\end{bmatrix}
\begin{smatrix}  S^T\\&I_{n-r}\end{smatrix}
\begin{smatrix}  {U}\\ 0_{(n-r)\times r} &I_{n-r}\end{smatrix} Q
$
shows that $\begin{bmatrix}  PLS & 0_{m\times(n-r)}\end{bmatrix}$ is
right-equivalent to $A$ and is therefore a column echelon form for
$A$. 
\end{proof}



%
Equivalently, the same reasoning applies for the recovery of a row echelon
form of $A$.
 Algorithm~\ref{alg:echelon} summarizes how a row and a column echelon form of $A$ can be
computed by sorting the first $r$ values of the permutations $\sigma_P$ and $\sigma_Q$.

\begin{algorithm}[htbp]
  \caption{Echelon forms from a PLUQ decomposition}
  \label{alg:echelon}
\begin{algorithmic}[1]
\Require{$P,L,U,Q$, a PLUQ decomposition of $A$ with $\RPM{A}=\Pi_{P,Q}$}
\Ensure{$C$: a column echelon form of $A$}
\Ensure{$R$: a row echelon form of $A$}
\State $(p_1,..,p_r) =\text{Sort}(\sigma_P(1),..,\sigma_P(r))$ \label{alg:line:sort}
\State $(q_1,..,q_r) =\text{Sort}(\sigma_Q(1),..,\sigma_Q(r))$ \label{alg:line:sort2}
\State $\tau = (\sigma_P^{-1}(p_1),..,\sigma_P^{-1}(p_r)) \in S_r$
\State $\chi = (\sigma_Q^{-1}(q_1),..,\sigma_Q^{-1}(q_r)) \in S_r $
\State $S_\tau = P(\tau)$
\State $S_\chi = P(\chi)^T$
\Return $C\leftarrow P  \begin{bmatrix}    L S_\tau & 0_{m\times (n-r)}  \end{bmatrix}$
and $R\leftarrow  
  \begin{smatrix} S_\chi U\\0_{(m-r)\times n}  \end{smatrix} Q$
\end{algorithmic}
\end{algorithm}

\begin{rem}
A row and column echelon form of the $i\times j$ leading
sub-matrix can be computed by removing rows of $PL$ below index $i$ and
filtering out the pivots of column index greater than $j$. 
The latter is achieved by replacing lines~\ref{alg:line:sort}
and~\ref{alg:line:sort2} in Algorithm~\ref{alg:echelon} by:
\newcounter{alglinesort}\setcounterref{alglinesort}{alg:line:sort}
\begin{algorithmic}[1]
  \makeatletter
  \setcounter{ALG@line}{-1}
  \addtocounter{ALG@line}{\thealglinesort}
  \makeatother
  \algrenewcommand\alglinenumber[1]{\footnotesize #1':}%
\State    $(p_1,..,p_s) = \text{Sort}( \{\sigma_P(k) : \sigma_Q(k)\leq j \text{
  and } 1\leq k\leq r\})$
\State    $(q_1,..,q_s) = \text{Sort}( \{\sigma_Q(k) : \sigma_P(k)\leq i \text{
  and } 1\leq k\leq r\})$.
\end{algorithmic}
\end{rem}

\subsection{The generalized Bruhat decomposition}\label{sec:XPY}

The generalization of the Bruhat decomposition for rank deficient matrices
of~\cite[Theorem~2.4]{Manthey:2007:Bruhat} is of the form $A=XFY$ where $X$ is in
column echelon form, $Y$ in row echelon form, and $F$ is a $r\times r$
permutation matrix.
Such a decomposition immediately follows from the echelon forms computed by
Algorithm~\ref{alg:echelon}.
Indeed, 
by Lemma~\ref{lem:echelon}, $PLS_\tau$ is in column echelon form and $S_\chi UQ$ is
in row echelon form. The identity
\begin{equation}\label{eq:XPY}
 A = \underbrace{\left(PLS_\tau\right)}_{X}
 \underbrace{S_\tau^TS_\chi^T}_{F} 
 \underbrace{\left(S_\chi{}UQ\right)}_{Y}
\end{equation}
then directly gives an $A=XFY$ decomposition.

In a similar way as in the full-rank case (see Section~\ref{ssec:LEU})
\cite[Theorem~2.4]{Manthey:2007:Bruhat} requires an additional condition on the
$X$ matrix to ensure the uniqueness of the decomposition: $F^TX_{\RRP(X),*}F$ must be lower
triangular, where $X_{I,*}$ denotes the lower triangular submatrix of $X$ formed
by the rows of its row rank profile.
This condition is not satisfied in general by the $X$ and $Y$ matrices computed
from Algorithm~\ref{alg:echelon} and a PLUQ decomposition revealing the rank
profile matrix. As a counter example, the matrix $\begin{smatrix}
  0&1\\1&a\end{smatrix}$ has two PLUQ decompositions revealing the rank profile
matrix $F=J_2=\begin{smatrix} 0& 1\\1&0\end{smatrix} $:
$A= J_2 I_2 \begin{smatrix}  1&a\\0&1\end{smatrix} I_2$ and $A=I_2
\begin{smatrix}  1&0\\a&1\end{smatrix} I_2 J_2$. For the latter one, $F^TXF =
\begin{smatrix} 1&a\\0&1\end{smatrix}$ which is not lower triangular.

In a same way as in equation~\eqref{eq:uniqueLEU}, we derive an equivalent
condition for the uniqueness of the generalized Bruhat decomposition.
We first need the following Lemma.
\begin{lem}
Let $C$ be an $m\times r$ matrix in column echelon form, and $P$ the $m\times r$
$r$-sub-permutation matrix with  ones at the location of the pivots of $C$.
Then 
$$
C_{\RRP(X),*} = P^T C.
$$
\end{lem}

As $P_{*,1..r}S_\tau$ is the $r$-sub-permutation matrix with ones on the
locations of the pivots of the column echelon form $X$, we deduce that
$$
F^T X_{\RRP(X),*}F = S_\chi S_\tau S_\tau^T (P_{*,1..r})^T PLS_\tau S_\tau^T
S_\chi^T = S_\chi L_{1..r,1..r} S_\chi^T.
$$

Again, a sufficient condition for this matrix to be lower triangular is that
$S_\chi=I_r$ which happens if the PLUQ decomposition has been computed using a
pivoting respecting the reverse lexicographic order.

\begin{cor}
  Let $A=PLUQ$ be a PLUQ decomposition computed with a pivoting strategy
  minimizing the reverse lexicographic order and performing row rotations. Then
  the $XFY$ decomposition of equation~\eqref{eq:XPY} is the unique generalized Bruhat
  decomposition as defined in~\cite[Theorem~2.4]{Manthey:2007:Bruhat}.
\end{cor}



\section{Improvement for low rank matrices}
\label{sec:lowrank}
\subsection{Storjohann and Yang's algorithm}
An alternative way to compute row and column rank profiles has recently been
proposed by Storjohann and Yang in~\cite{StoYan14,StoYan15}, reducing the time
complexity from a deterministic $O(mnr^{\omega-2})$ to a  Monte
Carlo probabilistic $2r^3+O(r^2(\log m + \log n) +mn)$ in~\cite{StoYan14} first,
and to $(r^\omega+mn)^{1+o(1)}$ in~\cite{StoYan15}, under the condition that the base field
contains at least $2\min(m,n)(1+\lceil\log_2 m\rceil+\lceil\log_2 n\rceil)$ elements. We show how these results
can be extended to the computation of the rank profile matrix.

The technique of Storjohann and Yang originates from the oracle system solving technique
of~\cite[\S 2]{MuSto00}. There, a linear system is solved or proved to be
inconsistent, by the incremental construction of a non-singular sub-matrix of
the augmented matrix of the system, and the computation of its inverse by rank one updates.

Applied to a right-hand-side vector $b$ sampled uniformly from the
column space of a matrix, this technique is used to incrementally build a list
of row indices $\mathcal{P} = [i_1,\dots,i_s]$ and of column indices
$\mathcal{Q} = [j_1,\dots j_s]$, for $s=1,\dots, r$, forming an incrementally
growing invertible sub-matrix $\submat{A}{\mathcal{P}}{\mathcal{Q}}$.
Each coordinate $(i_s,j_s)$ is that of a pivot as in a standard Gaussian
elimination algorithm, but the complexity improvement comes from the fact that
the elimination is only performed on this $s\times s$ submatrix, and not on the
linearly dependent rows and columns of the matrix $A$. To achieve this, 
the search for pivots is done by limiting the Schur complement computation to the
last column of the augmented system $\left[\begin{array}{c|c}  A&b \end{array}\right]$:
$$
b - \submat{A}{*}{\mathcal{Q}} \cdot (\submat{A}{\mathcal{P}}{\mathcal{Q}})^{-1} \cdot \subvec{b}{\mathcal{P}}.
$$
Under a
randomized assumption, the first non-zero component of this vector indicates the
value $i_s$ of the next linearly independent row. A similar approach is used for
determining the value of $j_s$.

As shown in~\cite[Theorem~6]{StoYan14}, the algorithm ends up
computing the row rank profile in $\mathcal{P}$  and the column rank profile in
$\mathcal{Q}$.
In fact, it computes all information of the rank profile matrix.

\begin{thm}
The $r$-permutation matrix whose ones are located at positions
$(i_s,j_s)$ for $1\leq s\leq r$ is precisely the rank profile matrix of $A$.
\end{thm}
\begin{pf}
  In this algorihtm, the pivot search looks for the first non-zero element in
  the Schur complement of $b$ which, with probability $1-1/\#\mathrm{K}$, will
  correspond to the first linearly independent row. 
  The exact Schur complement can then be computed on this row only:
$$
\submat{A}{i_s}{*} - \submat{A}{i_s}{\mathcal{Q}} \cdot (\submat{A}{\mathcal{P}}{\mathcal{Q}})^{-1}\cdot \submat{A}{\mathcal{P}}{*},
$$
so as to find the column position $j_s$ of the  first non zero element in it.
This search strategy minimizes the lexicographic order of the pivot coordinates.
No permutation strategy is being used, instead the columns where pivots have
already been found are zeroed out in the Schur complement computation, which has
the same  effect as a Column rotation. From the 4th line of
Table~\ref{tab:RPRSearchPerm}, this pivoting strategy reveals the rank profile matrix.
\end{pf}
The complexity of the direct Monte Carlo algorithm, in $O((m+n)r^2)$, hence
directly applies for the computation of the rank profile matrix. 
This complexity is reduced to $2r^3+O(r^2(\log n+\log m) +mn)$ by the
introduction of linear independence oracles in a divide and conquer 
scheme. Upon success of the probabilistic asumption, the values for
$\mathcal{P}$ and $\mathcal{Q}$ are unchanged.
Consequently, \cite[Theorem~19, Corollary~21]{StoYan14} also hold for the
computation of the rank profile matrix.
\begin{cor}
  There exist a Monte Carlo algorithm computing the rank profile matrix that has
  running time bounded by $2r^3+(r^2+m+n+|A|)^{1+o(1)}$, where $|A|$ denotes
  the number of non-zero elements in $A$.
\end{cor}

 The $r^3$ term in this complexity is from the iterative construction of
the inverses of the non-singular sub-matrices of order $s$ for $1\leq s\leq r$, by
rank one updates. To reduce this complexity to $O(r^\omega)$,
\cite{StoYan15} propose a relaxed computation of this
online matrix inversion.
In order to group arithmetic operations into matrix multiplications, their
approach is to anticipate part of the updates on the columns that will be appended
in the future. This however requires that one has an a priori knowledge on the
position of $r$ linearly independent columns of the initial matrix $A$ and, even further,
that the sub-matrix formed by these linearly independent columns has generic rank
profile.
The first condition is ensured by the Monte Carlo selection of linearly
independent columns of~\cite[Theorem 2.11]{CKL:13} and the second by the use of
a lower triangular Toeplitz preconditioner as in~\cite[Theorem 2]{KaSa:91}.
In the process, the row rank profile can be recovered, but all information on
the column rank profile is lost, and the rank profile matrix can thus not be recovered.

Yet, this algorithm can be run twice (once on $A$ and once on $A^T$), to recover
the row and the column rank profiles, and extract the corresponding  $r\times r$ invertible sub-matrix
$\submat{A}{\mathcal{P}}{\mathcal{Q}}$ of $A$. Following Corollary~\ref{cor:elu}, it then suffices to compute a PLU
decomposition of this sub-matrix, using an appropriate pivoting strategy, to
recover its rank profile matrix and therefore that of $A$ as well.

\begin{algorithm}
\begin{algorithmic}[1]
\caption{Low Rank Profile Matrix}
\label{alg:rpm:oracle}
\Require{$A$, an $m\times n$ matrix of rank $r$ over a field $\K$.}
\Ensure{$\RPM{A}$ the rank profile matrix of $A$ or FAILURE}
\State Compute the Row rank profile $\mathcal{P}=[i_1,\dots,i_r]$ of $A$ using~\cite{StoYan15}
\State Compute the Column rank profile $\mathcal{Q}=[j_1,\dots,j_r]$ of $A$
using~\cite{StoYan15} applied on $A^T$.
\State Let $B = \submat{A}{\mathcal{P}}{\mathcal{Q}}$
\State Compute $B = LUP$ a LUP decomposition of $B$ using a
  lexicographic pivot search and rotations for the column permutations.
\State  Return $\RPM{A} = P_\mathcal{P} \begin{smatrix} P &\\&0_{(m-r)\times (n-r))} \end{smatrix} P_\mathcal{Q}^T$
\end{algorithmic}
\end{algorithm}

\begin{thm}
  Algorithm~\ref{alg:rpm:oracle} is Monte Carlo probabilistic and computes the
  rank profile matrix of $A$ in time $(r^\omega+m+n+|A|)^{1+o(1)}$ field operations in $\K$. 
\end{thm}
Note that the last operation consists in the introduction of zero rows and
columns to expand the $r\times r$ permutation matrix $P$ into the $m\times n$
matrix $\RPM{A}$. The position of these zero rows and columns are deduced from the
row rank profile $\mathcal{P}$ and the column rank profile~$\mathcal{Q}$.

\subsection{Online LU decomposition}

The rank profile algorithms of Storjohann and Yang rely on an
online computation of a matrix inverse, either by rank one updates
in~\cite{StoYan14} or by a relaxed variant~\cite{StoYan15}, that corresponds to an online version
of the classic divide and conquer algorithm~\cite{Str69}.
We remark that these inversion algorithms can be advantageously replaced by LU
decomposition algorithms.
Following the argument in~\cite[Fig~2]{JPS:2013}, the use of LU decomposition
for linear system solving offers a speed-up factor of about 3, compared to the
use of a matrix inverse.
Indeed, solving a system from an LU decomposition is done by two
back substitutions, which has the same costs of $2n^2+O(n)$ 
field operations, as applying the inverse of the matrix to the right-hand-side
vector.  But the cost computing an LU decomposition is about three times as
fast as computing a
matrix inverse ($2/3n^3$ versus $2n^3$ when classic matrix arithmetic is used).
We present now how to replace the online matrix inverse by an online LU
decomposition in \cite{StoYan14,StoYan15}.

First consider the iterative construction  of the LU decomposition of
$\submat{A}{\mathcal{P}}{\mathcal{Q}}$ by rank one updates.
Suppose that $A_{s-1} = LU$ with $L$ lower triangular with a
unit diagonal, and $U$ upper triangular. Then we have 
\begin{equation}
\left[\begin{array}{c|c}
A_{s-1} & u \\
\hline 
v & d
\end{array}\right]
 = 
\left[\begin{array}{c|c} L \\ \hline vU^{-1} & 1 \end{array}\right] \cdot
\left[\begin{array}{c|c} U & L^{-1}u \\ \hline  & w \end{array}\right],
\label{eq:onlineLU}
\end{equation}
where $w=d-vU^{-1}L^{-1}u$.
This rank one update of the LU decomposition of the work matrix costs
$2s^2+O(s)$ to compute for an $(s-1)\times (s-1)$ matrix
$\submat{A}{\mathcal{P}}{\mathcal{Q}}$ (compared with $6s^2+O(s)$ in
\cite[Lemma~2]{StoYan14}). The remaining of the paper can then be used,
replacing every multiplication by the pre-computed matrix inverse $x\leftarrow
(\submat{A}{\mathcal{P}}{\mathcal{Q}})^{-1}\cdot b$ by 
two consecutive triangular system solving: $y\leftarrow L^{-1}b; x\leftarrow
U^{-1}y$.
This leads to a Monte Carlo algorithm computing the rank profile matrix in time $2/3r^3+(r^2+m+n+|A|)^{(1+o(1))}$.

\paragraph*{}{
Note that the decomposition in~\eqref{eq:onlineLU} uses two kinds of matrix
vector products: $vU^{-1}$ and $L^{-1}u$ contrary to the matrix inverse
formula~\eqref{eq:onlineInv} used in~\cite{StoYan14} and~\cite{StoYan15}
\begin{equation}
\left[\begin{array}{c|c}
A_{s-1} & u \\
\hline 
v & d
\end{array}\right]^{-1}
 = 
\left[\begin{array}{c|c} I_{s-1} &-A_{s-1}^{-1}u_s(d-v_sA_{s-1}^{-1}u_s)^{-1}\\
    \hline &(d-v_sA_{s-1}^{-1}u_s)^{-1} \end{array}\right] \cdot
\left[\begin{array}{c|c} I_{s-1} \\ \hline -v_s & 1 \end{array}\right] \cdot
\left[\begin{array}{c|c} A_{s-1}^{-1}  \\ \hline  & 1 \end{array}\right],
\label{eq:onlineInv}
\end{equation}
where only one matrix-vector product, $A_{s-1}^{-1} u_s$, appears.
The relaxation is achieved there by anticipating part of these products into larger
matrix-matrix multiplications. This requires that all of the column vectors that
will be treated are known in advance: this is ensured by the selection
of $r$ linearly independent columns of~\cite{CKL:13}.
Since no matrix-vector product involves row vectors $v_s$, it is still possible
to select and add them incrementally one after the other, thus allowing to
compute the row rank profile.

Now, relaxing the LU decomposition of equation~\eqref{eq:onlineLU} in a similar
way would require to also anticipate computations on the rows, namely $vU^{-1}$. This would be either too
expensive, if no pre-selection of $r$ linearly independent rows is performed, or
such a pre-selection would lose the row rank profile information that has to be
computed.
Hence we can not propose a similar improvement of the leading constant in the
relaxed case.
}










\begin{thebibliography}{36}
\providecommand{\natexlab}[1]{#1}
\providecommand{\url}[1]{\texttt{#1}}
\expandafter\ifx\csname urlstyle\endcsname\relax
  \providecommand{\doi}[1]{doi: #1}\else
  \providecommand{\doi}{doi: \begingroup \urlstyle{rm}\Url}\fi

\bibitem[Bhaskara~Rao(2002)]{Rao:2002}
K.P.S. Bhaskara~Rao.
\newblock \emph{The {Theory} of {Generalized} {Inverses} {Over} {Commutative}
  {Rings}}, volume~17 of \emph{Algebra {Logic} and {Application} series}.
\newblock Taylor and Francis, London, 2002.

\bibitem[Bourbaki(2008)]{Bourbaki:2008:lie}
N.~Bourbaki.
\newblock \emph{Groupes et Alg\`egres de {Lie}}.
\newblock Number Chapters~4--6 in Elements of mathematics. Springer, 2008.

\bibitem[Brown(1992)]{Brown:1992:matring}
William~C. Brown.
\newblock \emph{Matrices over Commutative Rings}.
\newblock Chapman \& Hall Pure and Applied Mathematics. CRC Press, 1992.
\newblock ISBN 9780824787554.

\bibitem[Brown(1998)]{Brown:1998:spanningrank}
William~C. Brown.
\newblock Null ideals and spanning ranks of matrices.
\newblock \emph{Communications in Algebra}, 26\penalty0 (8):\penalty0
  2401--2417, 1998.
\newblock \doi{10.1080/00927879808826285}.

\bibitem[Bruhat(1956)]{Bruhat:1956:Lie}
François Bruhat.
\newblock Sur les représentations induites des groupes de {Lie}.
\newblock \emph{Bulletin de la Société Mathématique de France}, 84:\penalty0
  97--205, 1956.
\newblock URL \url{http://eudml.org/doc/86911}.

\bibitem[Cheung et~al.(2013)Cheung, Kwok, and Lau]{CKL:13}
Ho~Yee Cheung, Tsz~Chiu Kwok, and Lap~Chi Lau.
\newblock Fast {Matrix} {Rank} {Algorithms} and {Applications}.
\newblock \emph{J. ACM}, 60\penalty0 (5):\penalty0 31:1--31:25, October 2013.
\newblock ISSN 0004-5411.
\newblock \doi{10.1145/2528404}.

\bibitem[Clark and Liang(1973)]{ClLi73}
W.~Edwin Clark and Joseph~J Liang.
\newblock Enumeration of finite commutative chain rings.
\newblock \emph{Journal of Algebra}, 27\penalty0 (3):\penalty0 445--453,
  December 1973.
\newblock ISSN 0021-8693.
\newblock \doi{10.1016/0021-8693(73)90055-0}.

\bibitem[Della~Dora(1973)]{delladora:1973:these}
Jean Della~Dora.
\newblock \emph{{Sur quelques algorithmes de recherche de valeurs propres}}.
\newblock Theses, {Universit{\'e} Joseph Fourier - Grenoble I}, July 1973.
\newblock URL \url{https://tel.archives-ouvertes.fr/tel-00010274}.

\bibitem[Dongarra et~al.(1998)Dongarra, Duff, Sorensen, and Vorst]{DDSV98}
Jack~J. Dongarra, Lain~S. Duff, Danny~C. Sorensen, and Henk A.~Vander Vorst.
\newblock \emph{Numerical Linear Algebra for High Performance Computers}.
\newblock SIAM, 1998.
\newblock ISBN 0898714281.

\bibitem[Dumas and Roch(2002)]{DuRo:2002}
Jean-Guillaume Dumas and Jean-Louis Roch.
\newblock On parallel block algorithms for exact triangularizations.
\newblock \emph{Parallel Computing}, 28\penalty0 (11):\penalty0 1531--1548,
  November 2002.
\newblock \doi{10.1016/S0167-8191(02)00161-8}.

\bibitem[Dumas et~al.(2008)Dumas, Giorgi, and Pernet]{jgd:2008:toms}
Jean-Guillaume Dumas, Pascal Giorgi, and Cl\'ement Pernet.
\newblock Dense linear algebra over prime fields.
\newblock \emph{ACM TOMS}, 35\penalty0 (3):\penalty0 1--42, November 2008.
\newblock \doi{10.1145/1391989.1391992}.

\bibitem[Dumas et~al.(2013)Dumas, Pernet, and Sultan]{DPS:2013}
Jean-Guillaume Dumas, Cl\'ement Pernet, and Ziad Sultan.
\newblock Simultaneous computation of the row and column rank profiles.
\newblock In Manuel Kauers, editor, \emph{Proc. ISSAC'13}, pages 181--188. ACM
  Press, 2013.
\newblock \doi{10.1145/2465506.2465517}.

\bibitem[Dumas et~al.(2014)Dumas, Gautier, Pernet, and Sultan]{DPS:2014}
Jean-Guillaume Dumas, Thierry Gautier, Clément Pernet, and Ziad Sultan.
\newblock Parallel computation of echelon forms.
\newblock In \emph{Euro-Par 2014 Parallel Proc.}, LNCS (8632), pages 499--510.
  Springer, 2014.
\newblock ISBN 978-3-319-09872-2.
\newblock \doi{10.1007/978-3-319-09873-9\_42}.

\bibitem[Dumas et~al.(2015{\natexlab{a}})Dumas, Pernet, and Sultan]{DPS15}
Jean-Guillaume Dumas, Cl\'ement Pernet, and Ziad Sultan.
\newblock Recursion based parallelization of exact dense linear algebra
  routines for {Gaussian} elimination.
\newblock \emph{Parallel Computing}, 2015{\natexlab{a}}.
\newblock \doi{10.1016/j.parco.2015.10.003}.

\bibitem[Dumas et~al.(2015{\natexlab{b}})Dumas, Pernet, and
  Sultan]{DPS:ELU:2015}
Jean-Guillaume Dumas, Cl{\'e}ment Pernet, and Ziad Sultan.
\newblock Computing the rank profile matrix.
\newblock In \emph{Proceedings of the 2015 ACM on International Symposium on
  Symbolic and Algebraic Computation}, ISSAC '15, pages 149--156, New York, NY,
  USA, 2015{\natexlab{b}}. ACM.
\newblock ISBN 978-1-4503-3435-8.
\newblock \doi{10.1145/2755996.2756682}.

\bibitem[Grigoriev(1981)]{Grigoriev:1981:bruhat}
Dima~Yu. Grigoriev.
\newblock Analogy of {Bruhat} decomposition for the closure of a cone of
  {Chevalley} group of a classical serie.
\newblock \emph{Soviet Mathematics Doklady}, 23\penalty0 (2):\penalty0
  393--397, 1981.

\bibitem[Hungerford(1968)]{hungerford1968}
Thomas~W. Hungerford.
\newblock On the structure of principal ideal rings.
\newblock \emph{Pacific Journal of Mathematics}, 25\penalty0 (3):\penalty0
  543--547, 1968.
\newblock URL \url{http://projecteuclid.org/euclid.pjm/1102986148}.

\bibitem[Ibarra et~al.(1982)Ibarra, Moran, and Hui]{IMH:1982}
Oscar~H. Ibarra, Shlomo Moran, and Roger Hui.
\newblock A generalization of the fast {LUP} matrix decomposition algorithm and
  applications.
\newblock \emph{J. of Algorithms}, 3\penalty0 (1):\penalty0 45--56, 1982.
\newblock \doi{10.1016/0196-6774(82)90007-4}.

\bibitem[Jeannerod et~al.(2013)Jeannerod, Pernet, and Storjohann]{JPS:2013}
Claude-Pierre Jeannerod, Cl{\'e}ment Pernet, and Arne Storjohann.
\newblock Rank-profile revealing {G}aussian elimination and the {CUP} matrix
  decomposition.
\newblock \emph{J. Symbolic Comput.}, 56:\penalty0 46--68, 2013.
\newblock ISSN 0747-7171.
\newblock \doi{10.1016/j.jsc.2013.04.004}.

\bibitem[Jeffrey(2010)]{Jeffrey:2010:lufact}
David~John Jeffrey.
\newblock {LU} factoring of non-invertible matrices.
\newblock \emph{ACM Comm. Comp. Algebra}, 44\penalty0 (1/2):\penalty0 1--8,
  July 2010.
\newblock ISSN 1932-2240.
\newblock \doi{10.1145/1838599.1838602}.

\bibitem[Kaltofen and Saunders(1991)]{KaSa:91}
Erich Kaltofen and B.~David Saunders.
\newblock On {W}iedemann's method of solving sparse linear systems.
\newblock In Harold~F. Mattson, Teo Mora, and T.~R.~N. Rao, editors,
  \emph{Applied {Algebra}, {Algebraic} {Algorithms} and {Error}-{Correcting}
  {Codes}}, number 539 in Lecture {Notes} in {Computer} {Science}, pages
  29--38. Springer Berlin Heidelberg, January 1991.
\newblock ISBN 978-3-540-54522-4 978-3-540-38436-6.
\newblock \doi{10.1007/3-540-54522-0\_93}.

\bibitem[Keller-Gehrig(1985)]{KG:1985}
Walter Keller-Gehrig.
\newblock Fast algorithms for the characteristic polynomial.
\newblock \emph{Th. Comp. Science}, 36:\penalty0 309--317, 1985.
\newblock \doi{10.1016/0304-3975(85)90049-0}.

\bibitem[Klimkowski and van~de Geijn(1995)]{KlvdGe95}
Kenneth Klimkowski and Robert~A. van~de Geijn.
\newblock Anatomy of a parallel out-of-core dense linear solver.
\newblock In \emph{International Conference on Parallel Processing}, volume~3,
  pages 29--33. CRC Press, aug 1995.

\bibitem[Krull(1938)]{Krull1938}
Wolfgang Krull.
\newblock Dimensionstheorie in stellenringen.
\newblock \emph{Journal für die reine und angewandte Mathematik},
  179:\penalty0 204--226, 1938.
\newblock URL \url{http://eudml.org/doc/150048}.

\bibitem[Malaschonok(2013)]{Malaschonok:2013:CASC}
Gennadi Malaschonok.
\newblock Generalized {Bruhat} decomposition in commutative domains.
\newblock In VladimirP. Gerdt, Wolfram Koepf, ErnstW. Mayr, and EvgeniiV.
  Vorozhtsov, editors, \emph{Computer Algebra in Scientific Computing}, volume
  8136 of \emph{Lecture Notes in Computer Science}, pages 231--242. Springer
  International Publishing, 2013.
\newblock ISBN 978-3-319-02296-3.
\newblock \doi{10.1007/978-3-319-02297-0\_20}.

\bibitem[Malaschonok(2010)]{Malaschonok:2010}
Gennadi~Ivanovich Malaschonok.
\newblock Fast generalized {Bruhat} decomposition.
\newblock In \emph{CASC'10}, volume 6244 of \emph{LNCS}, pages 194--202.
  Springer-Verlag, Berlin, Heidelberg, 2010.
\newblock ISBN 3-642-15273-2, 978-3-642-15273-3.
\newblock \doi{10.1007/978-3-642-15274-0\_16}.

\bibitem[Manthey and Helmke(2007)]{Manthey:2007:Bruhat}
Wilfried Manthey and Uwe Helmke.
\newblock {Bruhat} canonical form for linear systems.
\newblock \emph{Linear Algebra and its Applications}, 425\penalty0
  (2–3):\penalty0 261 -- 282, 2007.
\newblock ISSN 0024-3795.
\newblock \doi{10.1016/j.laa.2007.01.022}.
\newblock Special Issue in honor of Paul Fuhrmann.

\bibitem[McCoy(1948)]{McCoy:1948:rings}
Neal~H. McCoy.
\newblock \emph{Rings and ideals}.
\newblock Carus Monograph Series, no. 8. The Open Court Publishing Company,
  LaSalle, Ill., 1948.

\bibitem[Miller and Sturmfels(2005)]{MillerSturmfels:2005}
Ezra Miller and Bernd Sturmfels.
\newblock \emph{Combinatorial commutative algebra}, volume 227.
\newblock Springer, 2005.
\newblock \doi{10.1007/b138602}.

\bibitem[Mulders and Storjohann(2000)]{MuSto00}
Thom Mulders and Arne Storjohann.
\newblock Rational {Solutions} of {Singular} {Linear} {Systems}.
\newblock In \emph{Proceedings of the 2000 {International} {Symposium} on
  {Symbolic} and {Algebraic} {Computation}}, {ISSAC} '00, pages 242--249, New
  York, NY, USA, 2000. ACM.
\newblock ISBN 1-58113-218-2.
\newblock \doi{10.1145/345542.345644}.

\bibitem[Norton and Salagean(2000)]{Norton:2000:fcr}
Graham~H. Norton and Ana Salagean.
\newblock On the hamming distance of linear codes over a finite chain ring.
\newblock \emph{IEEE Trans. Information Theory}, 46\penalty0 (3):\penalty0
  1060--1067, 2000.
\newblock \doi{10.1109/18.841186}.

\bibitem[Storjohann(2000)]{Storjohann:2000:thesis}
Arne Storjohann.
\newblock \emph{Algorithms for Matrix Canonical Forms}.
\newblock PhD thesis, ETH-Zentrum, Z{\"u}rich, Switzerland, November 2000.

\bibitem[Storjohann and Yang(2014)]{StoYan14}
Arne Storjohann and Shiyun Yang.
\newblock Linear {Independence} {Oracles} and {Applications} to {Rectangular}
  and {Low} {Rank} {Linear} {Systems}.
\newblock In \emph{Proceedings of the 39th {International} {Symposium} on
  {Symbolic} and {Algebraic} {Computation}}, {ISSAC} '14, pages 381--388, New
  York, NY, USA, 2014. ACM.
\newblock ISBN 978-1-4503-2501-1.
\newblock \doi{10.1145/2608628.2608673}.

\bibitem[Storjohann and Yang(2015)]{StoYan15}
Arne Storjohann and Shiyun Yang.
\newblock A {Relaxed} {Algorithm} for {Online} {Matrix} {Inversion}.
\newblock In \emph{Proceedings of the 2015 {ACM} on {International} {Symposium}
  on {Symbolic} and {Algebraic} {Computation}}, {ISSAC} '15, pages 339--346,
  New York, NY, USA, 2015. ACM.
\newblock ISBN 978-1-4503-3435-8.
\newblock \doi{10.1145/2755996.2756672}.

\bibitem[Strassen(1969)]{Str69}
Prof~Volker Strassen.
\newblock {Gaussian} elimination is not optimal.
\newblock \emph{Numerische Mathematik}, 13\penalty0 (4):\penalty0 354--356,
  August 1969.
\newblock ISSN 0029-599X, 0945-3245.
\newblock \doi{10.1007/BF02165411}.

\bibitem[Tyrtyshnikov(1997)]{Tyrtyshnikov:1997:Bruhat}
E.~Tyrtyshnikov.
\newblock Matrix {Bruhat} decompositions with a remark on the {QR} ({GR})
  algorithm.
\newblock \emph{Linear Algebra and its Applications}, 250:\penalty0 61 -- 68,
  1997.
\newblock ISSN 0024-3795.
\newblock \doi{10.1016/0024-3795(95)00453-X}.

\end{thebibliography}



\end{document}